\documentclass[12pt, draftclsnofoot, onecolumn]{IEEEtran}
 \usepackage{amsmath,amssymb}
 \usepackage{subfigure}
 \usepackage{graphicx,graphics,color,psfrag}
 \usepackage{cite,balance}
 \usepackage{caption}
 \allowdisplaybreaks
 \usepackage{accents}
 \usepackage{amsthm}
 \usepackage{bm}
 \usepackage[english]{babel}
 \usepackage{multirow}
 \usepackage{enumerate}
 \usepackage{cases}
 \usepackage{stfloats}
 \usepackage{dsfont}
 \usepackage{color,soul}
 \usepackage{amsfonts}
 \usepackage{cite,graphicx,amsmath,amssymb}
 \usepackage{subfigure}
 \usepackage{fancyhdr}
 \usepackage{hhline}
 \usepackage{graphicx,graphics}
 \usepackage{array,color}
 \usepackage{amsmath}
 \usepackage{booktabs}
 \usepackage{algorithmic,algorithm}
 \newtheorem{theorem}{Theorem}
\include{header}

\begin{document}
\setlength{\topskip}{-3pt}

\newtheorem{lemma}{Lemma}
\newtheorem{proposition}{Proposition}
\newtheorem{remark}{Remark}

\title{\huge Integrating Sensing, Communication, and Power Transfer: Multiuser Beamforming Design}
\author{Ziqin Zhou, \emph{Student Member, IEEE}, Xiaoyang Li, \emph{Member, IEEE}, Guangxu Zhu, \emph{Member, IEEE}, Jie Xu, \emph{Senior Member, IEEE}, Kaibin Huang, \emph{Fellow, IEEE}, and Shuguang Cui, \emph{Fellow, IEEE}
\thanks{Part of this paper has been presented on the IEEE Wireless Communications and Networking Conference (WCNC), Glasgow, UK, 2023. \cite{li2023multi}

Ziqin Zhou, Xiaoyang Li, and Guangxu Zhu are with the Shenzhen Research Institute of Big Data, The Chinese University of Hong Kong-Shenzhen, Guangdong, China. Jie Xu and Shuguang Cui are with the Chinese University of Hong Kong, Shenzhen, China. Kaibin Huang is with The University of Hong Kong, Hong Kong. Corresponding author: Xiaoyang Li (lixiaoyang@sribd.cn).
}
}
\maketitle

\begin{abstract}
In the \emph{sixth-generation} (6G) networks, massive low-power devices are expected to sense environment and deliver tremendous data. To enhance the radio resource efficiency, the \emph{integrated sensing and communication} (ISAC) technique exploits the sensing and communication functionalities of signals, while the \emph{simultaneous wireless information and power transfer} (SWIPT) techniques utilizes the same signals as the carriers for both information and power delivery. The further combination of ISAC and SWIPT leads to the advanced technology namely \emph{integrated sensing, communication, and power transfer} (ISCPT). In this paper, a multi-user \emph{multiple-input multiple-output} (MIMO) ISCPT system is considered, where a base station equipped with multiple antennas transmits messages to multiple \emph{information receivers} (IRs), transfers power to multiple \emph{energy receivers} (ERs), and senses a target simultaneously. The sensing target can be regarded as a point or an extended surface. When the locations of IRs and ERs are separated, the MIMO beamforming designs are optimized to improve the sensing performance while meeting the communication and power transfer requirements. The resultant non-convex optimization problems are solved based on a series of techniques including Schur complement transformation and rank reduction. Moreover, when the IRs and ERs are co-located, the power splitting factors are jointly optimized together with the beamformers to balance the performance of communication and power transfer. To better understand the performance of ISCPT, the target positioning problem is further investigated. Simulations are conducted to verify the effectiveness of our proposed designs, which also reveal a performance tradeoff among sensing, communication, and power transfer.

\end{abstract}


\section{Introduction}
The \emph{sixth-generation} (6G) networks are expected to provide various sensing and communication services to support various emerging \emph{internet-of-things} (IoT) applications such as smart city and smart home \cite{saad2019vision}. Conventionally, the wireless sensing and communication systems are separately designed and may operate over different frequency bands to avoid mutual interference. To improve the spectrum efficiency and facilitate the data collection process, the \emph{integrated sensing and communication} (ISAC) technology utilizes the same frequency band to generate signals for both data transmission and radar sensing \cite{cui2021integrating}. In practice, the spectrum supporting ISAC ranges from the S-band (2-4 GHz) to the THz band \cite{liu2022integrated}. 

Wireless signals are not only the carriers of information but also the carriers of energy \cite{zhang2018wireless}. According to Koomey’s law, the reduction in power requirements of the electronics will result in explosive growth of low-power devices \cite{koomey2010implications}. For these low-power devices, \emph{wireless power transfer} (WPT) is expected to supply the energy via radio-frequency signals \cite{zeng2017communications}. By integrating the information and energy deliveries, \emph{simultaneous wireless information and power transfer} (SWIPT) has been proposed \cite{clerckx2018fundamentals}. However, existing literatures about WPT and SWIPT mainly focus on utilizing wireless signals as carriers of communication and power, while their sensing functionality has been overlooked.

In future 6G networks, the sensing, communication, and power transfer functionalities are expected to be integrated to enhance the radio resource efficiency and enable the data communication by massive low-power devices, which leads to the new research direction termed \emph{integrating sensing, communication, and power transfer} (ISCPT). The initial study on ISCPT has characterized the Pareto boundary of the sensing, communication, and power transfer functionalities for single \emph{information receiver} (IR) and single \emph{energy receiver} (ER) \cite{chen2022isac}. In current work, the particular \emph{multiple-input multiple-output} (MIMO) beamforming design for multiple IRs and multiple ERs is investigated to improve the sensing performance while guaranteeing the communication and power transfer requirements. The cases where the IRs and ERs are separated or co-located are both considered. For a co-located receiver, it harvests the energy and receives the information simultaneously. For the separated receivers, IR and ER are different devices with the former receiving the information and the latter harvesting the energy. 

The main contributions of this work are summarized below.
\begin{itemize}
\item \emph{Performance tradeoff among sensing, communication, and power transfer:} The performance tradeoff among sensing, communication, and power transfer in the scenario with multiple IRs and ERs are considered. According to the scale of the sensing target and its distance to the signal source, two different target models are considered, known as the point and extended targets, respectively. The sensing performance of the point target is evaluated by the \emph{Cram\'er-Rao Bound} (CRB) of the target parameters including the reflection coefficient and angel. The sensing performance of the extended target is evaluated by the \emph{mean squared error} (MSE) of the \emph{target response matrix} (TRM) estimation. The communication performance is evaluated by the \emph{signal-to-interference-plus-noise ratio} (SINR) of the IRs. The power transfer performance is evaluated by the harvested energy of the ERs. It is found that for both point target and extended target, increasing one of the performance metrics results in the deterioration of the others.

\item \emph{ISCPT beamforming design for the point target:} As for the point target, the beamforming is design for minimizing the sensing CRB while guaranteeing the energy requirements and the SINR requirements. To deal with the resultant non-convex optimization problem, the \emph{semidefinite relaxation} (SDR) and Schur complement transformation are applied. Subsequently, the solution approach based on \emph{rank reduction} (RR) method is proposed to derive a low-rank solution. Moreover, when the ERs and IRs are co-located, \emph{power splitting} (PS) method is adopted to divide the received signal for information decoding and \emph{energy harvesting} (EH), respectively. To balance the performance of communication and power transfer, the PS ratios of all signals are further considered as optimization variables together with the beamformers, which makes the problem more difficult to be solved. To deal with such a problem, auxiliary variables are introduced. Due to the existence of PS, the sensing performance is more susceptible to the varying energy and SINR requirements in the co-located case compared with the separated case.
     
\item \emph{ISCPT beamforming design for the extended target:} As for the extended target, the objective of the beamforming optimization problem is replaced by the sensing MSE. The solving approach in point target case becomes ineffective due to the existence of the inverse of the covariance matrix. Therefore, an alternative method is proposed to solve the problem by exploiting the rank property of the covariance matrix. The proposed method can effectively reduce the sensing error under the same constraints compared with the conventional scheme based on eigenmode decomposition. Moreover, the energy and SINR requirements as well as the amount of IRs further affect the sensing performance via the beamforming design.

\item \emph{Target positioning based on ISCPT:} To quantify the performance of ISCPT design, the use case of target positioning is investigated. Specifically, the target location is estimated based on the information extracted from the echo signals. It can be observed that the ISCPT signals can achieve accurate target positioning while guaranteeing the low SINR requirements at the IRs and low EH requirements at the ERs, while the performance of target positioning degrades when the SINR or EH requirement becomes high.
\end{itemize}

The remainder of this paper is organized as follows. Section II provides a comprehensive review of the existing literatures about ISAC, SWIPT, and target positioning. The system models of ISCPT are specified in Section III. The beamforming designs for the point target and the extended target estimations are provided by Section IV and Section V, respectively. The target positioning based on ISCPT is illustrated in Section VI. Simulations are conducted in Section VII to verify the performance of our proposed schemes, followed by a conclusion of this paper in Section VIII.


\section{Related Works}
\subsection{ISAC and Positioning}
The inception of ISAC can be traced back to the pioneering efforts in integrating radar and communication systems, wherein data was encoded within a cluster of radar pulses \cite{liu2022survey}. From the viewpoint of information theory, the rate distortion theory was applied to unifying the radar and communication performance \cite{chiriyath2015inner}. The particular waveform design for ISAC has been proposed in \cite{liu2018toward} to support simultaneous information transmission and target detection. Furthermore, the ISAC design was developed for \emph{multiple-input multiple-output} (MIMO) systems \cite{hua2023mimo}. The beamformers were designed to improve the ISAC efficiency by exploiting the spatial diversity \cite{hua2023optimal}. Taking the multi-user communication into account, two MIMO beamforming designs for ISAC were proposed in \cite{liu2018mu}, namely shared and separated designs.

The advantages of spectrum sharing have made ISAC a widely adopted technology in various systems, including \emph{reconfigurable intelligent surface} (RIS) systems \cite{wang2021joint}, edge learning systems \cite{zhang2022accelerating}, millimeter-wave systems \cite{kumari2017ieee}, smart homes \cite{huang2020joint}, \emph{Internet of Things} (IoT) \cite{li2023over}, vehicular networks \cite{yuan2020bayesian}, and \emph{unmanned aerial vehicle} (UAV) systems \cite{lyu2022joint}. To mitigate multi-user interference in ISAC, a joint design of waveform and discrete phase shift was conducted based on the deployment of RIS \cite{wang2021joint}. ISAC was also applied to accelerate the edge learning process by designing wireless signals for dual purposes of dataset generation and uploading \cite{zhang2022accelerating}. In millimeter-wave band, an IEEE 802.11ad-based radar was deployed to support both automotive radar functionality and communication network \cite{kumari2017ieee}. In smart homes, traditional sensing devices were endowed with communication capabilities, while the sensing ability of WiFi signals was enhanced \cite{huang2020joint}. The sensing and communication capabilities of mobile devices were further exploited for supporting the IoT applications \cite{li2023over}. In vehicular networks, wireless sensing capabilities were utilized to obtain vehicle states and facilitate communication \cite{yuan2020bayesian}. In UAV-enabled ISAC systems, the joint optimization of maneuver and beamforming designs were investigated for simultaneous communication with multiple users and sensing of potential targets \cite{lyu2022joint}. To facilitate the data collection process, ISAC is further integrated with over-the-air computation \cite{li2023integrated}.

As a vital application scenario of sensing, target positioning has been widely investigated in a series of literatures. The new radio positioning method was introduced in \cite{parkvall20205g} based on the use of a location server, where the positioning can be based on downlink or uplink. A new visible light communications based indoor positioning system was proposed in \cite{lin2017experimental}. The authors of \cite{bai2019camera} proposed a novel camera assisted received signal strength ratio positioning algorithm. In \cite{zhu2017three}, a novel positioning framework based on the angle differences of arrival in three-dimensional coordinate systems was introduced. An improved location tracking algorithm was designed in \cite{lee2006improved} with velocity estimation in cellular radio networks. In \cite{wang2016csi}, a novel deep-learning-based indoor fingerprinting system using channel state information was presented. 

Despite the rich literature on ISAC, the power supply for ISAC devices is largely overlooked. Moreover, the radio signals are not only the carriers of information but also energy. Therefore, it is natural to integrate the operations of sensing, communication, and power transfer.

\subsection{WPT and SWIPT}
Originally intended for point-to-point power transmission, the WPT technology has incorporated cutting-edge methods from electromagnetic radiation and inductive coupling \cite{zeng2017communications}. The practical experiments have verified that WPT can charge a variety of low-power wearables and mobile devices, including electronic watches, IoT sensors, hearing aids, and wireless keyboards \cite{lu2015wireless}. In particular, a \emph{radio-frequency} (RF) circuit has been developed in \cite{jabbar2010rf} to power mobile devices in urban networks with high density of RF sources. To power mobiles in cellular networks, a RF power beacon was further deployed in \cite{huang2014enabling}. As a promising replacement of the charging cables, WPT has been widely adopted to power the devices in diversified scenarios, including the mobile edge computing \cite{you2016energy}, fast data aggregation \cite{li2019wirelessly}, mobile crowd sensing \cite{li2018wirelessly}, and ISAC \cite{li2022wirelessly}. 


Integrating the information and power transfer further gave rise to the emerging field of SWIPT \cite{clerckx2018fundamentals}. In SWIPT, energy and information are simultaneously delivered from one or multiple transmitter(s) to one or multiple receiver(s) . To improve the spatial efficiency, MIMO SWIPT has been investigated in \cite{zhang2013mimo} with one IR and one ER, where the transmit waveform was designed to balance the performance of communication and power transfer. The waveform design was further extended to account for the case with multiple IRs and ERs in \cite{xu2014multiuser}. A series of studies focused on applying SWIPT to various communication systems and networks, including two-way transmissions \cite{popovski2013interactive}, MIMO communications \cite{park2013joint}, and cognitive networking~\cite{ng2016multiobjective}. 


Despite the wide applications of WPT and SWIPT in communication systems, the incorporation of data sensing is seldom investigated, which deserves to be explored.


\subsection{ISCPT}
The combination of ISAC and SWIPT leads to the emerging technology namely ISCPT. In \cite{li2022wirelessly}, a power beacon was applied for delivering energy to multiple devices and thus enabling their sensing and communication capabilities. In \cite{chen2022isac}, the fundamental performance tradeoff among sensing, communication, and power transfer was unveiled in a MIMO system, where the access point simultaneously senses a target, delivers message to an IR and energy to an ER. The work of \cite{li2023multi} further considered the scenario including multiple IRs and ERs. The extension to multiple sensing targets was investigated in \cite{zeng2022beamforming}. Compared with the initial studies, this paper provides a more elaborative illustration of the ISCPT system for specific types of the sensing targets as well as the locations of IRs and ERs.

\begin{figure}[t]
  \centering
  \subfigure[]{
  \label{FigSeparated}
  \includegraphics[scale=0.6]{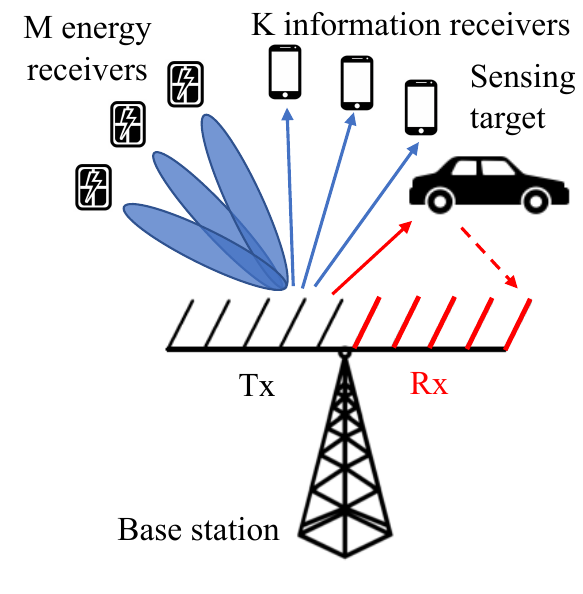}}
  \subfigure[]{
  \label{FigColocated}
  \includegraphics[scale=0.6]{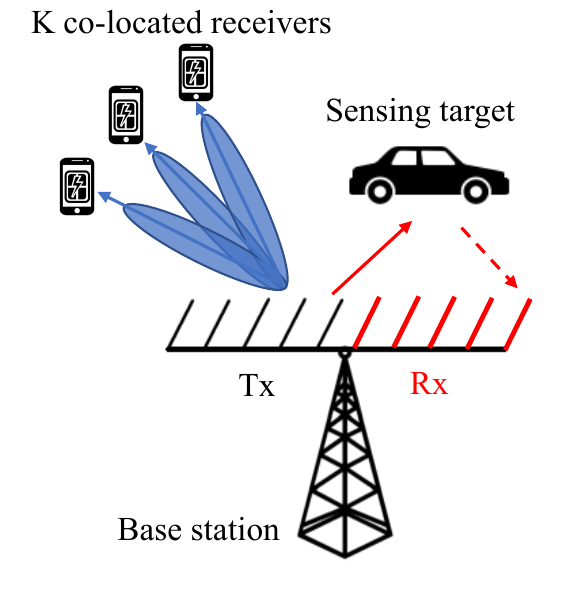}}
  \caption{ISCPT system with a) separated; b) co-located IRs and ERs}
  \label{FigSys}
\end{figure}

\section{System Model}
Consider an ISCPT system comprising one \emph{base station} (BS) equipped with $N_t$ transmit antennas and $N_r$ receive antennas to sense one target, deliver information to $K < N_t$ single-antenna IRs and wirelessly charge $M < N_t$ single-antenna ERs. As depicted in Fig.~\ref{FigSys}, the locations of ERs can be different from or the same as those of IRs, which is known as the separated or co-located case. The specific models of transmitted and received signals are elaborated below.

\subsection{Model of Signals Transmitted by the BS}
The signal transmitted by the BS over the $T \geq N_t$ symbol intervals is denoted as $\bold{X} \in \mathbb{C}^{N_t \times T}$ and can be expressed as
\begin{equation} \label{eq:Tx}
\bold{X} = \bold{W}_D \bold{S}, 
\end{equation}
where $\bold{W}_D = [\bold{w}_1,\bold{w}_2,...,\bold{w}_K] \in \mathbb{C}^{N_t \times K}$ is the beamforming matrix with $K$ columns being the beamformers for the IRs and wirelessly charge ERs simultaneously. $\bold{S} \in \mathbb{C}^{K \times T}$ contains $K$ unit-power data streams intended for $K$ IRs. The signal streams are assumed to be independent of each other, i.e.,
\begin{equation}
\mathbb{E}[\bold{S}\bold{S}^H] = \bold{I}_{K},
\end{equation} 
where $\bold{I}_{K}$ is an identity matrix with dimension $K$. The transmit power can be written as
\begin{equation}
\mathbb{E}[\text{tr}(\bold{X}\bold{X}^H)] = \text{tr}(\bold{W}_D \bold{W}_D^H) = \text{tr}\left(\sum_{k=1}^{K} \bold{w}_k \bold{w}_k^H\right).
\end{equation} 

\subsection{Model of Radar Sensing}
At the BS, the echo signal reflected by the extended target is denoted by $\bold{Y}_R \in \mathbb{C}^{N_r \times T}$, which can be expressed as 
\begin{equation}
\bold{Y}_R = \bold{G} \bold{X} + \bold{N}_R, 
\end{equation} 
where $\bold{N}_R$ is the \emph{additive white Gaussian noise} (AWGN) matrix with the variance of each entry being $\sigma_R^2$, and $\bold{G} \in \mathbb{C}^{N_r \times N_t}$ is the TRM between the BS and the target, which can be of different forms depending on the type of targets. We first consider the case where the target is regarded as a point, while the case of extended target is investigated in Section V. The TRM of a point target is specified as
\begin{equation}
\bold{G} = \alpha \bold{A}(\theta), 
\end{equation} 
where $\alpha$ represents the reflection coefficient, $\theta$ is the azimuth angle of the target relative to the BS, and $\bold{A}(\theta) = \bold{b}(\theta) \bold{a}^H(\theta)$ is formed by multiplexing the steering vector $\bold{a}(\theta) \in \mathbb{C}^{N_t \times 1}$ of the transmit antennas and $\bold{b}(\theta) \in \mathbb{C}^{N_r \times 1}$ of the receive antennas. The transmit and receive antennas are assumed to be \emph{uniform linear array} (ULA) with half-wavelength antenna spacing. Therefore, we have \begin{equation}
    \bold{a}(\theta) = [e^{-j \frac{N_{t} - 1}{2} \pi \sin\theta} ,e^{-j \frac{N_{t} - 3}{2} \pi \sin\theta},...,e^{j \frac{N_{t} - 1}{2} \pi \sin\theta}]^T,
    \end{equation}
    \begin{equation}
\bold{b}(\theta) = [e^{-j \frac{N_{r} - 1}{2} \pi \sin\theta} ,e^{-j \frac{N_{r} - 3}{2} \pi \sin\theta},...,e^{j \frac{N_{r} - 1}{2} \pi\sin\theta}]^T.
    \end{equation}
According to \cite{liu2021cramer}, the CRB for estimating angle $\theta$ is expressed as
\begin{equation}
\text{CRB}(\theta) = \frac{\sigma_R^2 \text{tr}(\bold{A}^H(\theta)\bold{A}(\theta)\bold{R}_X)}{2|\alpha|^2 T (\text{tr}(\dot{\bold{A}}^H(\theta)\dot{\bold{A}}(\theta)\bold{R}_X)\text{tr}(\bold{A}^H(\theta)\bold{A}(\theta)\bold{R}_X)-|\text{tr}(\dot{\bold{A}}^H(\theta)\bold{A}(\theta)\bold{R}_X)|^2)}, 
\end{equation} 
where $\dot{\bold{A}}(\theta) = \frac{\partial \bold{A}(\theta)}{\partial \theta}$, and $\bold{R}_X = \frac{1}{T} \bold{X} \bold{X}^H = \frac{1}{T} \bold{W}_D \bold{S} \bold{S}^H \bold{W}_D^H$. When the transmission duration $T$ becomes sufficiently long, we have $\bold{R}_X = \bold{W}_D \bold{W}_D^H = \sum_{k=1}^{K} \bold{w}_k \bold{w}_k^H$.

\subsection{Model of Signals Received by the ERs and IRs}
The model of received signals has different forms depending on the locations of receivers. In particular, we will discuss the following two models:
\begin{itemize}
\item {\bf Separated ERs and IRs}: In this case, the locations of ERs are different from those of IRs. The received signal matrix at $K$ IRs is denoted by $\bold{Y}_C \in \mathbb{C}^{K \times T}$, which is expressed as
\begin{equation}
\bold{Y}_C = \bold{H} \bold{X} + \bold{N}_C, 
\end{equation} 
where $\bold{N}_C$ is the AWGN matrix with the variance of each entry being $\sigma_C^2$ and $\bold{H} = [\bold{h}_1,\bold{h}_2,...,\bold{h}_K]^H \in \mathbb{C}^{K \times N_t}$ is the communication channel matrix, which is assumed to be perfectly known at the BS. For the $k$-th IR, the communication performance is measured by its receive SINR denoted by 
\begin{equation}
\gamma_k = \frac{|\bold{h}_k^H\bold{w}_k|^2}{\sum_{i=1,i \neq k}^{K} |\bold{h}_k^H\bold{w}_i|^2 + \sigma_C^2}.
\end{equation}
The signal received at the $m$-th ER is denoted by $\bold{y}_m \in \mathbb{C}^{1 \times T}$, which can be expressed as
\begin{equation}
\bold{y}_m = \bold{c}_m \bold{X} + \bold{n}_m, 
\end{equation} 
where $\bold{c}_m \in \mathbb{C}^{1 \times N_t}$ is the channel between the BS and the $m$-th ER, $\bold{n}_m \in \mathbb{C}^{1 \times T}$ is the AWGN vector. Due to the negligible power of noise compared with those of signals, the received RF energy can be expressed as\footnote{In practice, the ER uses a rectifier to convert the received RF to the \emph{direct-current} (DC) signals for charging the battery. Although the RF-to-DC conversion is a non-linear process, the harvested DC power is generally a monotonically increasing function with respect to the received RF power. Therefore, the EH constraints on DC power can be equivalently rewritten as constraints on RF power.}
\begin{equation}
E_m = \beta_m \mathbb{E}[|\bold{c}_m \bold{X}|^2] T = \beta_m \bold{c}_m \left(\sum_{k=1}^{K}\bold{w}_k\bold{w}_k^H\right)\bold{c}_m^H T,
\end{equation}
where $\beta_m \in [0,1]$ is the EH coefficient. 

\item {\bf Co-located ERs and IRs}: In this case, the IR and ER are colocated, i.e.,  there are $M = K$ ERs co-located with IRs. A power splitter is employed to split the received RF power into two portions for information decoding and energy harvesting. Let $\rho_k \in [0,1]$ denote the portion of signal power split for EH, the energy harvested by the $k$-th receiver is expressed as
\begin{equation}
E_k = \rho_k \beta_k \mathbb{E}[|\bold{h}_k \bold{X}|^2] T  = \rho_k \beta_k \bold{h}_k^H \left(\sum_{k=1}^{K}\bold{w}_k\bold{w}_k^H\right)\bold{h}_k T.
\end{equation}
Therefore, the portion of signal power split for information receiving at the $k$-th receiver is $(1 - \rho_k)$, and thus the SINR can be written as
\begin{equation}
\gamma_k = \frac{(1 - \rho_k)|\bold{h}_k^H\bold{w}_k|^2}{(1 - \rho_k)\sum_{i=1,i \neq k}^{K} |\bold{h}_k^H\bold{w}_i|^2 + \sigma_C^2}.
\end{equation}
\end{itemize}





\section{ISCPT Beamforming Design for Point Target}
In this section, the ISCPT beamforming designs for point target sensing are investigated with respect to both the models of separated and co-located ERs and IRs. Therefore, the objective of the beamforming optimization problem is to minimize the sensing CRB. The specific problem formulations and solution approaches are discussed in the following sub-sections.

\subsection{Sensing CRB Minimization for Separated ERs and IRs}
Given the maximum transmit power denoted by $P$, the required SINR level for the IRs denoted by $\{\eta_k\}$ and the required harvested energy for the ERs denoted by $\{Q_m\}$, the sensing CRB minimization problem can be formulated as
\begin{subequations}
\begin{align}
\textbf{(P1)}\quad \min_{\{\bold{w}_k\}_{k=1}^K} & \text{CRB}(\theta) \\ 
\text{s.t.} \quad
& \frac{|\bold{h}_k^H\bold{w}_k|^2}{\sum_{i=1,i \neq k}^{K} |\bold{h}_k^H\bold{w}_i|^2 + \sigma_C^2} \geq \eta_k, k=1,...,K,\label{Eq:P1b}\\ 
 & \beta_m \bold{c}_m \left(\sum_{k=1}^{K}\bold{w}_k\bold{w}_k^H\right)\bold{c}_m^H \geq Q_m, m=1,...,M, \label{Eq:P1c}\\
& \text{tr}\left(\sum_{k=1}^{K}\bold{w}_k\bold{w}_k^H\right) \leq P. \label{Eq:P1d}
\end{align}
\end{subequations}
The constraints \eqref{Eq:P1b} represent the required SINRs for IRs, while those in \eqref{Eq:P1c} correspond to the EH requirements for ERs. Additionally, \eqref{Eq:P1d} denotes the transmit power constraint.
It is observed that problem (P1) is non-convex due to the fractional structure in \eqref{Eq:P1b}. According to the Schur complement condition, problem (P1) can be equivalently transformed into  
\begin{subequations}
\begin{align}
\textbf{(Q1)} \quad \min_{t, \{\bold{w}_k\}_{k=1}^K} & -t \\ 
\text{s.t.} \quad
& \left[ 
\begin{array}{c}
\text{tr}(\dot{\bold{A}}^H\dot{\bold{A}}\sum_{k=1}^{K} \bold{w}_k \bold{w}_k^H) - t \qquad \text{tr}(\dot{\bold{A}}^H\bold{A}\sum_{k=1}^{K} \bold{w}_k \bold{w}_k^H)\\
\text{tr}(\dot{\bold{A}}^H\bold{A}\sum_{k=1}^{K}\bold{w}_k \bold{w}_k^H) \qquad \quad \text{tr}(\bold{A}^H\bold{A}\sum_{k=1}^{K} \bold{w}_k \bold{w}_k^H) \\
\end{array}
\right] \succeq \bold{0}, \label{Eq:P2b}\\ 
 & \frac{|\bold{h}_k^H\bold{w}_k|^2}{\sum_{i=1,i \neq k}^{K}|\bold{h}_k^H\bold{w}_i|^2 + \sigma_C^2} \geq \eta_k, k=1,...,K,\label{Eq:Q1c}\\ 
& \beta_m \bold{c}_m \left(\sum_{k=1}^{K}\bold{w}_k\bold{w}_k^H\right)\bold{c}_m^H \geq Q_m, m=1,...,M, \label{Eq:Q1d}\\
& \text{tr}\left(\sum_{k=1}^{K}\bold{w}_k\bold{w}_k^H\right) \leq P, \label{Eq:Q1e}
\end{align}
\end{subequations}
where $\bold{A} \triangleq \bold{A}(\theta)$ and $\dot{\bold{A}} \triangleq \dot{\bold{A}}(\theta)$. As problem (Q1) is still non-convex, the SDR technique is applied to relax the problem as a convex one. To this end, define $\bold{H}_k = \bold{h}_k \bold{h}_k^H$, $\bold{C}_m = \bold{c}_m^H \bold{c}_m$, and beamfomer covariance matrixs $\bold{W}_k = \bold{w}_k \bold{w}_k^H$, where the desired solution requires $\text{rank}(\bold{W}_k) = 1$ and $\bold{W}_k \succeq \bold{0}$. However, due to the existence of rank constraints, problem (Q1) remains unsolvable. After removing the rank constraints, the problem (Q1) can be relaxed as:
\begin{subequations}
\begin{align}
\textbf{(C1)} \quad \min_{t, \{\bold{W}_k\}_{k=1}^K} & -t \\ 
\text{s.t.} \quad
& \left[ 
\begin{array}{c}
 \text{tr}(\dot{\bold{A}}^H\dot{\bold{A}}\sum_{k=1}^{K} \bold{W}_k) - t \qquad  \text{tr}(\dot{\bold{A}}^H\bold{A} \sum_{k=1}^{K} \bold{W}_k)\\
 \text{tr}(\dot{\bold{A}}^H\bold{A}\sum_{k=1}^{K} \bold{W}_k) \qquad \quad \text{tr}(\bold{A}^H\bold{A}\sum_{k=1}^{K} \bold{W}_k) \\
\end{array}
\right] \succeq \bold{0}, \label{Eq:P3b}\\ 
& (1+\eta_k)\text{tr}(\bold{H}_k\bold{W}_k) - \eta_k \text{tr}\left(\bold{H}_k\sum_{i=1}^{K} \bold{W}_i\right) \geq \eta_k \sigma_C^2, k=1,...,K, \label{Eq:C1c}\\ 
& \beta_m  \text{tr} \left(\bold{C}_m \sum_{k=1}^{K} \bold{W}_k\right) \geq Q_m, m=1,...,M, \label{Eq:C1d}\\
& \text{tr}\left(\sum_{k=1}^{K} \bold{W}_k\right) \leq P, \label{Eq:P3e}\\
& \bold{W}_k \succeq \bold{0}, k = 1,...,K.\label{Eq:C1f}
\end{align}
\end{subequations}
The SDP may lead to high-rank solutions, thus making them infeasible for the original problem (C1). Therefore, we derive low-rank solution by using the rank-reduction technique.


First, let us define the dual variables for problem (C1), which are $\{z_1, z_2, ..., z_{K+1}\}$ and $\{v_1, v_2, ..., v_M\}$ that are associated with $K + M + 1$ linear constraints, and $\{\bold{Z}_1, \bold{Z}_2, ..., \bold{Z}_{K+1}\} \succeq \bold{0}$ that are associated with $K + 1$ semidefinite constraints. By assuming that the optimality is reached with $\{z_1, z_2, ..., z_{K+1}\}$, $\{v_1, v_2, ..., v_M\}$, $\{\bold{Z}_1, \bold{Z}_2, ... , \bold{Z}_{K+1}\}$, and $\{\bold{W}_1, \bold{W}_2, ..., \bold{W}_{K}\}$, the following complementary conditions hold true:
\begin{subequations}
\begin{align}
&\left((1+\eta_k)\text{tr}(\bold{H}_k\bold{W}_k) - \eta_k  \text{tr}(\bold{H}_k\sum_{i=1}^{K} \bold{W}_i) - \eta_k \sigma_C^2\right)z_k = 0, z_k \leq 0, k = 1,...,K,\\
&\left(\beta_m \text{tr} (\bold{C}_m \sum_{k=1}^{K} \bold{W}_k) - Q_m\right)v_m = 0, v_m \leq 0, m =1,...,M,\\
&\left(\text{tr}(\sum_{k=1}^{K} \bold{W}_k) - P\right) z_{K+1} = 0, z_{K+1} \geq 0,\\
&\bold{W}_k\bold{Z}_k = \bold{0}, \bold{Z}_k \succeq \bold{0}, k = 1,...,K,\\
&\left[ 
\begin{array}{c}
\text{tr}(\dot{\bold{A}}^H\dot{\bold{A}}\sum_{k=1}^{K} \bold{W}_k) - t \qquad \text{tr}(\dot{\bold{A}}^H\bold{A}\sum_{k=1}^{K} \bold{W}_k)\\
 \text{tr}(\dot{\bold{A}}^H\bold{A}\sum_{k=1}^{K} \bold{W}_k) \qquad \quad  \text{tr}(\bold{A}^H\bold{A}\sum_{k=1}^{K} \bold{W}_k) \\
\end{array}
\right] 
\bold{Z}_{K+1} = \bold{0}, \bold{Z}_{K+1} \succeq \bold{0}.
\end{align}       
\end{subequations}
Let us denote $R_k = \text{rank}(\bold{W}_k), k = 1,...,K$. By decomposing $\bold{W}_k = \bold{V}_k \bold{V}_k^H$ with $\bold{V}_k \in \mathbb{C}^{N_t \times R_k}$ (e.g., via low-rank Cholesky decomposition), one can get $\text{tr}(\bold{\Omega}\bold{W}_k) = \text{tr}(\bold{V}_k^H \bold{\Omega} \bold{V}_k)$ , where $\bold{\Omega}$ is an arbitrarily given matrix. Consider the following linear equations
\begin{subequations}
\begin{align}
&\sum_{k=1}^{K} \text{tr}(\bold{V}_k^H \dot{\bold{A}}^H\dot{\bold{A}} \bold{V}_k \bold{\Delta}_k) = 0,\label{le1}\\
&\sum_{k=1}^{K} \text{tr}(\bold{V}_k^H \bold{A}^H \bold{A} \bold{V}_k \bold{\Delta}_k) = 0,\label{le2}\\
&\sum_{k=1}^{K} \text{tr}(\bold{V}_k^H \dot{\bold{A}}^H \bold{A} \bold{V}_k \bold{\Delta}_k) = 0,\label{le3}\\
&\sum_{k=1}^{K} \text{tr}(\bold{V}_k^H \bold{V}_k \bold{\Delta}_k) = 0,\label{le4}\\
&\text{tr}(\bold{V}_k^H\bold{H}_k\bold{V}_k\bold{\Delta}_k) - \eta_k \sum_{i=1,i \neq k}^{K} \text{tr}(\bold{V}_i^H\bold{H}_k\bold{V}_i \bold{\Delta}_i) = 0, k = 1,...,K,\label{le5}\\
&\beta_m \sum_{k=1}^{K} \text{tr}(\bold{V}_k^H\bold{C}_m\bold{V}_k\bold{\Delta}_k) = 0, m = 1,...,M, \label{le6}
\end{align}
\end{subequations}
where $\bold{\Delta}_k$ is a $R_k \times R_k$ matrix. The system consists of $\sum_{k=1}^{K} R_k^2$ real-valued unknown variables and $K+M+4$ linear equations. The linear equations will have a solution if $\sum_{k=1}^{K} R_k^2 \geq K+M+4$. Let us denote the eigenvalues of $\bold{\Delta}_k$ as $\delta_{k1}, ..., \delta_{kR_k}$, we define
\begin{equation}
\delta_{\max} = \arg\max_{\delta_{kl}} \{|\delta_{kl}|, k = 1,...K, l = 1,...,R_k\}.
\end{equation}
We then update the solution by
\begin{equation}
\bold{W}'_k = \bold{V}_k (\bold{I}_{R_k} - \frac{1}{\delta_{\max}} \bold{\Delta}_k)\bold{V}_k^H, k = 1,...,K.
\end{equation}
The remaining task is to prove that $\{\bold{W}'_k\}_{k=1}^{K}$ is still a solution to the original equations with a reduced rank.


\begin{itemize}
\item {\bf Rank reduction}: It is obvious that
\begin{equation}\label{Eq:Reduction}
\sum_{k=1}^{K} \text{rank}(\bold{W}_k') \leq \sum_{k=1}^{K} \text{rank}(\bold{W}_k) - 1
\end{equation}
due to the choice of $\delta_{\max}$, i.e., the rank is reduced at least by one.

\item {\bf Primal Feasibility}: It is easy to verify that $\{\bold{W}_k'\}_{k=1}^{K}$ still satisfy the constraints, since the newly added terms in the constraints are all zero, thanks to the fact that $\{\bold{\Delta}_i\}_{i=1}^{K}$ is the solution of the linear system. We then define
\begin{equation}
\left[ 
\begin{array}{c}
\sum_{k=1}^{K} \text{tr}(\dot{\bold{A}}^H\dot{\bold{A}}\bold{W}_k) - t \qquad \sum_{k=1}^{K} \text{tr}(\dot{\bold{A}}^H\bold{A}\bold{W}_k)\\
\sum_{k=1}^{K} \text{tr}(\dot{\bold{A}}^H\bold{A}\bold{W}_k) \qquad \quad \sum_{k=1}^{K} \text{tr}(\bold{A}^H\bold{A}\bold{W}_k) \\
\end{array}
\right] \triangleq 
\left[ 
\begin{array}{c}
a-t \qquad c\\
c \qquad \quad b
\end{array}
\right] \succeq \bold{0},
\end{equation}
where $a,b \geq 0$. When the optimality is reached, we have $a - t - |c|^2b^{-1} = 0$, or equivalently
\begin{equation}
\left[ 
\begin{array}{c}
a-t \qquad c\\
c \qquad \quad b
\end{array}
\right] = \left[ 
\begin{array}{c}
|c|^2b^{-1} \qquad c\\
c \qquad \quad~~ b
\end{array}
\right].
\end{equation}
By substituting $\{\bold{W}_k'\}_{k=1}^{K}$ into (23), $a$, $b$, and $c$ keep unchanged. As a result, the objective value $-t$ remains the same.

\item {\bf Dual feasibility and complementarity}: Denote the updated dual variables as $\{v_1',v_2',...,v_M'\}$, $\{z_1',z_2',...,z_{K+1}'\}$, and $\{\bold{Z}_1',\bold{Z}_2',...,\bold{Z}_{K+1}'\}$. Note that by keeping $z_k'=z_k, k=1,...,K$, the complementary conditions for SINR constraints are satisfied. By keeping $v_m'=v_m, m=1,...,M$, the complementary conditions for harvested energy constraints are satisfied. Moreover, by letting $\bold{Z}_k' = \bold{Z}_k, k = 1,...,K$, we have
\begin{equation}
\text{tr}(\bold{W}_k' \bold{Z}_k') \!=\! \text{tr}(\bold{W}_k \bold{Z}_k) \!-\! \frac{1}{\delta_{\max}}\text{tr}(\bold{V}_k\bold{\Delta}_k\bold{V}_k^H \bold{Z}_k) \!=\! - \frac{1}{\delta_{\max}}\text{tr}(\bold{V}_k\bold{\Delta}_k\bold{V}_k^H \bold{Z}_k),k=1,...,K.
\end{equation}
Given the positive semi-definiteness of $\bold{W}_k$ and $\bold{Z}_k$, we have
\begin{equation}
\bold{W}_k \bold{Z}_k = 0 \Leftrightarrow \text{tr}(\bold{W}_k \bold{Z}_k) = \text{tr}(\bold{V}_k^H \bold{Z}_k\bold{V}_k) = 0 \Leftrightarrow \bold{V}_k^H \bold{Z}_k\bold{V}_k = 0.
\end{equation}
Hence,
\begin{equation}
\text{tr}(\bold{W}_k' \bold{Z}_k') = 0 \Leftrightarrow \bold{W}_k' \bold{Z}_k' = 0, k=1,...,K,
\end{equation}
which suggests that the complementary conditions still hold for $\bold{Z}_k' = \bold{Z}_k$ and $\bold{W}_k'$, $k =1,2,...,K$. Finally, let us deal with $\bold{Z}_{K+1}'$. Upon letting 
\begin{equation}
\bold{Z}_{K+1} = \left[ 
\begin{array}{c}
|c|^{-2}b \qquad -c^{-1}\\
-c^{-1} \qquad \quad~~ b^{-1}
\end{array}
\right] \succeq \bold{0}, 
\end{equation}
we have
\begin{equation}
\left[ 
\begin{array}{c}
a-t \qquad c\\
c \qquad \quad b
\end{array}
\right] \bold{Z}_{K+1} = \bold{0}.
\end{equation}
Since substituting $\{\bold{W}_k'\}_{k=1}^{K}$ into (23) does not change the values of $a$, $b$ and $c$, one may simply let $\bold{Z}_{K+1}' = \bold{Z}_{K+1}$, in which case the complementary condition with respect to $\bold{Z}_{K+1}'$ and the semidefinite constraint (23) still hold, where the optimal $t$ keeps unchanged.
\end{itemize}

Based on the above, it is clear that $\{\bold{W}_k'\}_{k=1}^{K}$ is an optimal solution to problem (C1) with reduced rank. We then check if $\sum_{k=1}^{K} R_k^2 \geq K+M+4$. If so, repeat the above rank-reduction procedure until $\sum_{k=1}^{K} R_k^2 \leq K+M+4$. Moreover, it is obvious that $\text{rank}(\bold{W}_k) \geq 1$ for $k = 1, 2, ..., K$. Therefore, there always exists a solution satisfying
\begin{equation}
K \leq \sum_{k=1}^{K} \text{rank}^2(\bold{W}_k) \leq K+M+4.
\end{equation}

After deriving the low-rank solution, the solution of the original problem can then be extracted by applying the Cholesky decomposition or eigenvalue decomposition.

\subsection{Sensing CRB Minimization for Co-located ERs and IRs}
As for the co-located ERs and ERs, the sensing CRB minimization problem can be formulated as
\begin{subequations}
\begin{align}
\textbf{(P2)} \quad \min_{\{\bold{w}_k\}_{k=1}^K,\{\rho_k\}_{k=1}^K} & \text{CRB}(\theta) \\ 
\text{s.t.} \quad~~~
& \frac{(1 - \rho_k)|\bold{h}_k^H\bold{w}_k|^2}{(1 - \rho_k)\sum_{i=1,i \neq k}^{K} |\bold{h}_k^H\bold{w}_i|^2 + \sigma_C^2} \geq \eta_k, k=1,...,K,\label{Eq:P2b}\\ 
 & \rho_k \beta_k \bold{h}_k^H \left(\sum_{k=1}^{K}\bold{w}_k\bold{w}_k^H\right)\bold{h}_k T \geq Q_k, k=1,...,K, \label{Eq:P2c}\\
& \text{tr}\left(\sum_{k=1}^{K}\bold{w}_k\bold{w}_k^H\right) \leq P. \label{Eq:P2d}
\end{align}
\end{subequations}

Note that in \textbf{(P2)}, the $K$ receivers serve dual roles as both information and energy receivers. The PS factors $\{\rho_1,\rho_2,...,\rho_K\}$ are introduced as optimization variables to determine the appropriate ratio of the received signal for information decoding and EH.  After applying the SDR, one can get:
\begin{subequations}
\begin{align}
\textbf{(C2)} ~ \min_{t, \{\bold{W}_k\}_{k=1}^K,\{\rho_k\}_{k=1}^K} & -t \\ 
 \text{s.t.} \qquad~
& \left[ 
\begin{array}{c}
 \text{tr}(\dot{\bold{A}}^H\dot{\bold{A}} \sum_{k=1}^K\bold{W}_{k} ) - t \qquad \text{tr}(\dot{\bold{A}}^H\bold{A} \sum_{k=1}^K\bold{W}_{k})\\
 \text{tr}(\dot{\bold{A}}^H\bold{A}\sum_{k=1}^K\bold{W}_{k}) \qquad \quad  \text{tr}(\bold{A}^H\bold{A}\sum_{k=1}^K\bold{W}_{k}) \\
\end{array}
\right] \succeq \bold{0}, \label{Eq:C2b}\\ 
& (1-\rho_k)\left((1+\eta_k)\text{tr}(\bold{H}_k\bold{W}_k) - \eta_k \text{tr}\left(\bold{H}_k\sum_{i=1}^K\bold{W}_{i}\right)\right) \geq \eta_k \sigma_c^2, k=1,...,K, \label{Eq:C2c}\\ 
& \rho_k \beta_k  \text{tr} \left(\bold{H}_k \sum_{i=1}^K\bold{W}_{i}\right) \geq Q_k, k =1,...,K, \label{Eq:C2d}\\
& \text{tr}\left(\sum_{k=1}^K\bold{W}_{k}\right) \leq P, \label{Eq:C2e}\\
& \bold{W}_k \succeq \bold{0}, k = 1,...,K. \label{Eq:C2f}
\end{align}
\end{subequations}
For $k = 1, 2, ..., K$, let us define the auxiliary variables $\frac{\eta_k \sigma_c ^2}{1-\rho_k} \leq c_k$ and $\frac{Q_k}{\rho_k\beta_m} \leq d_k$. Then, problem (C2) is expressed as:
\begin{subequations}
\begin{align}
\textbf{(Q2)} \quad \min_{t, \{\bold{W}_k\}_{k=1}^K, \{\rho_k\}_{k=1}^K}  & -t \\ 
\text{s.t.} \quad~~~
& \left[ 
\begin{array}{c}
\text{tr}(\dot{\bold{A}}^H\dot{\bold{A}} \sum_{k=1}^K\bold{W}_{k}) - t \qquad  \text{tr}(\dot{\bold{A}}^H\bold{A} \sum_{k=1}^K\bold{W}_{k})\\
 \text{tr}(\dot{\bold{A}}^H\bold{A}\sum_{k=1}^K\bold{W}_{k}) \qquad \quad \text{tr}(\bold{A}^H\bold{A}\sum_{k=1}^K\bold{W}_{k}) \\
\end{array}
\right] \succeq \bold{0}, \label{Eq:Q2b}\\ 
& (1+\eta_k)\text{tr}(\bold{H}_k\bold{W}_k) - \eta_k \text{tr}\left(\bold{H}_k\sum_{i=1}^K\bold{W}_{i}\right) \geq c_k, k=1,...,K, \label{Eq:Q2c}\\ 
& \text{tr} \left(\bold{H}_k \sum_{i=1}^K\bold{W}_{i}\right) \geq d_k, k =1,...,K, \label{Eq:C3d}\\
& \text{tr}\left(\sum_{k=1}^K\bold{W}_{k}\right) \leq P, \label{Eq:Q2e}\\
& \frac{\eta_k \sigma_c ^2}{1-\rho_k} \leq c_k, k=1,...,K, \label{Eq:Q2f}\\
& \frac{Q_k}{\rho_k\beta_m} \leq d_k, k=1,...,K, \label{Eq:Q2g}\\
& \bold{W}_k \succeq \bold{0}, k=1,...,K.\label{Eq:Q2h}
\end{align}
\end{subequations}
Adopting the Schur complement lemma, the constraints \eqref{Eq:Q2f} and \eqref{Eq:Q2g} are reformulated as
\begin{equation}
   \left[ \begin{array}{c}
c_k \qquad \sqrt{\eta_k}\sigma_c\\
\sqrt{\eta_k}\sigma_c \qquad 1-\rho_k\\ 
\end{array}\right] \succeq \bold{0},
\end{equation}
\begin{equation}
\left[ \begin{array}{c}
d_k \qquad \sqrt{\frac{Q_k}{\beta_k}}\\\
\sqrt{\frac{Q_k}{\beta_k}} \qquad \rho_k\\ 
\end{array}\right] \succeq \bold{0}.
\end{equation}

Then problem (\textbf{Q2}) is rewritten as 
\begin{subequations}
\begin{align}
\textbf{(Z2)} ~ \min_{t,\{\bold{W}_k\}_{k=1}^K, \{\rho_k\}_{k=1}^K, \{c_k\}_{k=1}^K, \{d_k\}_{k=1}^K}  &-t \\ 
\text{s.t.} \qquad \qquad~~~
& \left[ 
\begin{array}{c}
 \text{tr}(\dot{\bold{A}}^H\dot{\bold{A}} \sum_{k=1}^K\bold{W}_{k}) - t \qquad  \text{tr}(\dot{\bold{A}}^H\bold{A} \sum_{k=1}^K\bold{W}_{k})\\
\text{tr}(\dot{\bold{A}}^H\bold{A}\sum_{k=1}^K\bold{W}_{k}) \qquad \quad  \text{tr}(\bold{A}^H\bold{A}\sum_{k=1}^K\bold{W}_{k}) \\
\end{array}
\right] \succeq \bold{0}, \label{Eq:C2b}\\ 
& (1+\eta_k)\text{tr}(\bold{H}_k\bold{W}_k) - \eta_k \text{tr}\left(\bold{H}_k\sum_{i=1}^K\bold{W}_{i}\right) \geq c_k, k=1,...,K, \label{Eq:C2c}\\ 
& \text{tr} \left(\bold{H}_k \sum_{i=1}^K\bold{W}_{i}\right) \geq d_k, k =1,...,K, \label{Eq:C3d}\\
& \text{tr}\left(\sum_{k=1}^K\bold{W}_{k}\right) \leq P, \label{Eq:C3e}\\
&    \left[ \begin{array}{c}
c_k \qquad \sqrt{\eta_k}\sigma_c\\
\sqrt{\eta_k}\sigma_c \qquad 1-\rho_k\\ 
\end{array}\right] \succeq \bold{0}, k =1,...,K, \\
& \left[ \begin{array}{c}
d_k \qquad \sqrt{\frac{Q_k}{\beta_k}}\\\
\sqrt{\frac{Q_k}{\beta_k}} \qquad \rho_k\\ 
\end{array}\right] \succeq \bold{0}, k =1,...,K, \\
& \bold{W}_k \succeq \bold{0}, k =1,...,K, \label{Eq:C3f}
\end{align}
\end{subequations}
Problem (Z2) is a standard SDP and can be solved via numerical tools like CVX. The similar techniques for Section IV. A can be applied to approaching the optimal solution of the original problem.

\section{ISCPT Beamforming Design for Extended Target}
In this section, the beamforming designs are further investigated for the extended target, where the target is regarded as a surface with a large number of distributed point-like scatterers. The corresponding TRM can be specified as
\begin{equation}
\bold{G} = \sum_{n=1}^{N_s} \alpha_n \bold{A}(\theta_n), 
\end{equation} 
where $N_s$ represents the number of scatterers, $\alpha_n$ and $\theta_n$ represent the reflection coefficient and the azimuth angle of the $n$-th scatterer. Suppose that the BS has no prior knowledge about the number $N_s$ and the corresponding angles of scatterers, as the echoes can be randomly reflected. Therefore, we resort to estimating the complete TRM $\bold{G}$ instead of the angles $\theta_n$. According to \cite{li2023integrated}, the \emph{maximum likelihood estimation} (MLE) of $\bold{G}$ is given as
\begin{equation} \label{MLEG}
\hat{\bold{G}} = \bold{Y}_R \bold{X}^H(\bold{X} \bold{X}^H)^{-1}.
\end{equation} 
The corresponding MSE for estimating $\bold{G}$ is computed as
\begin{equation} \label{MSEG}
\text{MSE}(\bold{G}) = \mathbb{E}\left\{\|\bold{G} - \bold{\hat{G}}\|^2\right\} = \frac{N_r \sigma_R^2}{T}\text{tr}\left\{\bold{R}_X^{-1}\right\},
\end{equation} 
It should be noted that as the MLE of $\bold{G}$ is simply a linear estimation problem in the presence of the \emph{independent and identically distributed} (i.i.d.) Gaussian noise, the MSE in \eqref{MSEG} is equal to the CRB.

According to \eqref{eq:Tx}, $\bold{X} \in \mathbb{C}^{N_t \times T}$ is rank-deficient, since
\begin{equation} 
\text{rank}(\bold{X}) \leq \min \{\text{rank}(\bold{W}_D), \text{rank}(\bold{S})\} = K < N_t \leq T.
\end{equation} 
Consequently, if we transmit only $K$ signal streams, the available \emph{degrees of freedom} (DoF) are not enough to recover the rank-$N_t$ matrix $\bold{G}$. To deal with this problem, we introduce an auxiliary structure to $\bold{X}$ to maximize its DoFs up to $N_t$. This is achieved by transmitting dedicated probing streams alongside the data streams intended for $K$ users. Specifically, we define the beamforming matrix $\bold{W}_D = [\bold{w}_1, \bold{w}_2, ..., \bold{w}_K, \bold{w}_{K+1}, ..., \bold{w}_{K+N_t}] \in \mathbb{C}^{N_t \times (K+N_t)}$, where the first $K$ columns represent the beamformers for the IR signals, while the remaining columns are auxiliary beamformers for radar sensing. Furthermore, we define $\bold{S} = [\bold{S}_C^T, \bold{S}_A^T]^T$ with $\mathbb{E}[\bold{S}\bold{S}^H] = \bold{I}_{K+N_t}$, where $\bold{S}_C \in \mathbb{C}^{K \times T}$ consists of $K$ unit-power data streams intended for the $K$ IRs, and $\bold{S}_A \in \mathbb{C}^{N_t \times T}$ consists of $N_t$ unit-power auxiliary streams. According to the large-number law, $\bold{R}_X = \frac{1}{T} \bold{X} \bold{X}^H = \frac{1}{T} \bold{W}_D \bold{S} \bold{S}^H \bold{W}_D^H = \sum_{k=1}^{K+N_t} \bold{w}_k \bold{w}_k^H$.
The corresponding transmit power can be written as
\begin{equation}
\mathbb{E}[\text{tr}(\bold{X}\bold{X}^H)] = \text{tr}\left(\sum_{k=1}^{K+N_t} \bold{w}_k \bold{w}_k^H\right).
\end{equation} 

For the separated case, the SINR of signal received by the $k$-th IR is expressed as
\begin{equation}
\gamma_k = \frac{|\bold{h}_k^H\bold{w}_k|^2}{\sum_{i=1,i \neq k}^{K+N_t} |\bold{h}_k^H\bold{w}_i|^2 + \sigma_C^2}.
\end{equation}
The energy harvested by the $m$-th ER is expressed as
\begin{equation}
E_m = \beta_m \mathbb{E}[|\bold{c}_m \bold{X}|^2] T = \beta_m \bold{c}_m \left(\sum_{k=1}^{K+N_t}\bold{w}_k\bold{w}_k^H\right)\bold{c}_m^H T.
\end{equation}

For the co-located case, the SINR of signal received by the $k$-th receiver is expressed as
\begin{equation}
\gamma_k = \frac{(1 - \rho_k)|\bold{h}_k^H\bold{w}_k|^2}{(1 - \rho_k)\sum_{i=1,i \neq k}^{K+N_t} |\bold{h}_k^H\bold{w}_i|^2 + \sigma_C^2}.
\end{equation}
The energy harvested by the $k$-th receiver is expressed as
\begin{equation}
E_k = \rho_k \beta_k \mathbb{E}[|\bold{h}_k \bold{X}|^2] T  = \rho_k \beta_k \bold{h}_k^H \left(\sum_{k=1}^{K+N_t}\bold{w}_k\bold{w}_k^H\right)\bold{h}_k T.
\end{equation}

In the following sub-sections, the beamformers are optimized to minimize the sensing MSE of the extended target while guaranteeing the SINR and EH requirements with respect to both the cases with separated and co-located ERs and IRs.

\subsection{Sensing MSE Minimization for Separated ERs and IRs}
For the scenario with separated ERs and IRs, the problem of sensing MSE minimization can be formulated as follows:\begin{subequations}
\begin{align}
\textbf{(P3)} \quad \min_{\{\bold{w}_k\}_{k=1}^{K+N_t}} & \frac{N_r \sigma_R^2}{T}\text{tr}\left\{\left(\sum_{k=1}^{K+N_t} \bold{w}_k \bold{w}_k^H\right)^{-1}\right\} \\ 
\text{s.t.} \quad~~
& \frac{|\bold{h}_k^H\bold{w}_k|^2}{\sum_{i=1,i \neq k}^{K+N_t} |\bold{h}_k^H\bold{w}_i|^2 + \sigma_C^2} \geq \eta_k, k = 1,...,K,\label{Eq:P3b}\\ 
& \beta_m \bold{c}_m \left(\sum_{k=1}^{K+N_t}\bold{w}_k\bold{w}_k^H\right)\bold{c}_m^H \geq Q_m, m = 1,...,M, \label{Eq:P3c}\\
& \text{tr}\left(\sum_{k=1}^{K+N_t}\bold{w}_k\bold{w}_k^H\right) \leq P. \label{Eq:P3d}
\end{align}
\end{subequations}
Due to the existence of $\{\bold{w}_{K+1},...,\bold{w}_{K+N_t}\}$, problem (P3) has more variables to be optimized than problem (P1). Fortunately, one can observe that the effect of the variables $\{\bold{w}_K,...,\bold{w}_{K+N_t}\}$ on problem (P3) is only reflected by $\sum_{k=K+1}^{K+N_t} \bold{w}_k\bold{w}_k^H$. Therefore, these variables can be replaced by $\bold{R}_X = \sum_{k=1}^{K+N_t} \bold{w}_k\bold{w}_k^H$. Let $\bold{H}_k = \bold{h}_k \bold{h}_k^H$, $\bold{C}_m = \bold{c}_m^H \bold{c}_m$, and $\bold{W}_k = \bold{w}_k \bold{w}_k^H$, where $\text{rank}(\bold{W}_k) = 1$ and $\bold{W}_k \succeq \bold{0}$. By relaxing the rank-1 constraints, problem (P3) can be simplified as
\begin{subequations}
\begin{align}
\textbf{(Q3)} \quad \min_{\bold{R}_X, \{\bold{W}_k\}_{k=1}^K} & \frac{N_r \sigma_R^2}{T}\text{tr}\left\{\bold{R}_X^{-1}\right\} \\ 
\text{s.t.} \quad~~
& (1+\eta_k)\text{tr}(\bold{H}_k\bold{W}_k) - \eta_k \text{tr}(\bold{H}_k\bold{R}_X) \geq \eta_k \sigma_C^2, k = 1,...,K, \label{Eq:Q3c}\\ 
& \beta_m  \text{tr} \left(\bold{C}_m \bold{R}_X\right) \geq Q_m, m = 1,...,M, \label{Eq:Q3d}\\
& \text{tr}\left(\bold{R}_X\right) \leq P, \label{Eq:Q3e}\\
& \bold{W}_k \succeq \bold{0}, k = 1,...,K, \label{Eq:Q3f}\\
&\bold{R}_X - \sum_{k=1}^{K} \bold{W}_k  \succeq \bold{0}. \label{Eq:Q3g}
\end{align}
\end{subequations}
Note that the objective function is convex with respect to $\bold{R}_X$, and both $\eqref{Eq:Q3c}$ and $\eqref{Eq:Q3d}$ are convex. Therefore, \textbf{(Q3)} is a convex problem that can be efficiently solved using dedicated tools such as Matlab CVX. By denoting the optimal solution of \textbf{(Q3)} as $\bar{\bold{R}}_X$ and ${\bold{\bar{W}}_1,...,\bold{\bar{W}}_{K}}$, one can observe that rank$(\bold{R}_X) = N_t$ and rank$(\bold{\bar{W}}_k)\geq 1$. To obtain the rank-one solution of $\bold{W}_k$ while maintaining optimality, a solving approach is proposed based on the following theorem.

\begin{theorem}[]Given optimal solution $\bar{\bold{R}}_X$ and ${\bold{\bar{W}}_1,...,\bold{\bar{W}}_{K}}$ to problem (Q3), the constructed $\widetilde{\bold{R}}_X$ and ${\bold{\widetilde{W}}_1,...,\bold{\widetilde{W}}_{K}}$ in the following are also optimal.
\begin{equation}
    \widetilde{\bold{R}}_X = \bar{\bold{R}}_X, \qquad {\bold{\widetilde{W}}}_k = \frac{\bold{\bar{W}}_k \bold{H}_k \bold{\bar{W}}_k^H}{\rm{tr}(\bold{H}_k\bold{\bar{W}}_k )},
\end{equation}
where $\rm{rank}(\bold{\widetilde{W}}_k)=1, k=1,2,...,K$.
\end{theorem}
\begin{proof}
It can be observed that rank$(\bold{\widetilde{W}}_k)= 1$ since rank$(\bold{H}_k)= 1$. The next step is to demonstrate that the solution is both feasible and optimal for \textbf{(Q3)}. It is worth noting that the optimality is due to $\widetilde{\bold{R}}_X = \bar{\bold{R}}_X$.

\item {\bf Feasibility}:
Note that for the SINR constraint, one has
\begin{equation} \label{SINR-trans}
    \text{tr}(\bold{H}_k\bold{\widetilde{W}}_i) = \bold{h}_k^H \bold{\bar{W}}_i \bold{h}_k \bold{h}_k^H \bold{\bar{W}}_i^H \bold{h}_k(\bold{h}_k^H \bold{\bar{W}}_i^H \bold{h}_k)^{-1} = \bold{h}_k^H \bold{\bar{W}}_i^H \bold{h}_k = \text{tr}(\bold{H}_k\bold{\bar{W}}_i), i =1,...,K+N_t.
\end{equation}
By substituting \eqref{SINR-trans} into \eqref{Eq:Q3c}, 
\begin{equation}
(1+\eta_k)\text{tr}(\bold{H}_k\bold{\widetilde{W}}_k) - \eta_k \sum_{i=1}^{N_t}\text{tr}(\bold{H}_k\bold{\widetilde{W}}_i) = (1+\eta_k)\text{tr}(\bold{H}_k\bold{\bar{W}}_k) - \eta_k \sum_{i=1}^{N_t}\text{tr}(\bold{H}_k\bold{\bar{W}}_i) \geq \eta_k \sigma_C^2. 
\end{equation}
The feasibility of \eqref{Eq:Q3g} holds since $ \bold{\bar{W}}_k-\bold{\widetilde{W}}_k \succeq \bold{0}$. 
For any $\bold{v} \in \mathbb{C}^{N_t \times 1 }$, 
\begin{equation}\label{43}
    \bold{v}^H(\bold{\bar{W}}_k-\bold{\widetilde{W}}_k)\bold{v} =\bold{v}^H\bold{\bar{W}}_k\bold{v}-(\bold{h}_k^H \bold{\bar{W}}_k\bold{h}_k)^{-1} (\bold{v}^H \bold{\bar{W}}_k\bold{h}_k)^2.
\end{equation}
By applying Cauchy-Schwarz inequality, 
\begin{equation}
    (\bold{h}_k^H \bold{\bar{W}}_k\bold{h}_k)(\bold{v}^H \bold{\bar{W}}_k\bold{v}) \geq (\bold{v}^H\bold{\bar{W}}_k\bold{h}_k)^2.
\end{equation}
Hence, \eqref{43} must be non-negative, and $\bold{\bar{W}}_k-\bold{\widetilde{W}}_k$ is semi-positive definite. Since the remaining constraints can also be satisfied, the solution is both feasible and optimal for \textbf{(Q3)}, thereby completing the proof.
\end{proof}


\subsection{Sensing MSE Minimization for Co-located ERs and IRs}
For the scenario with co-located ERs and IRs, the problem of sensing MSE minimization can be formulated as follows:
\begin{subequations}
\begin{align}
\textbf{(P4)} \quad \min_{\{\bold{w}_k\}_{k=1}^{K+N_t},\{\rho_k\}_{k=1}^K} & \frac{N_r \sigma_R^2}{T}\text{tr}\left\{\left(\sum_{k=1}^{K+N_t} \bold{w}_k \bold{w}_k^H\right)^{-1}\right\} \\ 
\text{s.t.} \qquad~
& \frac{(1 - \rho_k)|\bold{h}_k^H\bold{w}_k|^2}{(1 - \rho_k)\sum_{i=1,i \neq k}^{K+N_t} |\bold{h}_k^H\bold{w}_i|^2 + \sigma_C^2} \geq \eta_k, k = 1,...,K,\label{Eq:P4b}\\ 
& \rho_k \beta_k \bold{h}_k \left(\sum_{k=1}^{K+N_t}\bold{w}_k\bold{w}_k^H\right)\bold{h}_k^H T \geq Q_k, k = 1,...,K, \label{Eq:P4c}\\
& \text{tr}\left(\sum_{k=1}^{K+N_t}\bold{w}_k\bold{w}_k^H\right) \leq P. \label{Eq:P4d}
\end{align}
\end{subequations}
Following the similar relaxation approach of problem (P3), the problem (P4) can be converted to:
\begin{subequations}
\begin{align}
\textbf{(C4)} \quad \min_{\bold{R}_X,\{\bold{W}_k\}_{k=1}^K,\{\rho_k\}_{k=1}^K} & \frac{N_r \sigma_R^2}{T}\text{tr}\left\{\bold{R}_X^{-1}\right\} \\ 
\text{s.t.} \qquad~~~
& (1-\rho_k)((1+\eta_k)\text{tr}(\bold{H}_k\bold{W}_k) - \eta_k \text{tr}(\bold{H}_k\bold{R}_X)) \geq \eta_k \sigma_C^2, k = 1,...,K, \label{Eq:C4c}\\ 
& \rho_k\beta_m  \text{tr} \left(\bold{C}_m \bold{R}_X\right) \geq Q_k, k = 1,...,K, \label{Eq:C4d}\\
& \text{tr}\left(\bold{R}_X\right) \leq P, \label{Eq:C4e}\\
& \bold{W}_k \succeq \bold{0}, k = 1,...,K,\label{Eq:C4f}\\
&\bold{R}_X - \sum_{k=1}^{K} \bold{W}_k  \succeq \bold{0}. \label{Eq:C4g}
\end{align}
\end{subequations}
Combining the analysis for the co-located case in Section IV.B with the MSE minimization method discussed in Section V.B, the problem \textbf{(C4)} can be solved using similar techniques.

\section{Case Study on Target Positioning via ISCPT}
In this section, the ISCPT scheme is applied for the use case of point target positioning. As shown in Fig.~\ref{FigPosition}, the BS is responsible for positioning the point target based on the extracted information such as angle and distance from the reflected radar signals. For the target with TRM $\bold{G} = \alpha \bold{A}(\theta)$, let $a_{pq}(\theta)$ represents the element located in the $p$-th row and $q$-th column of $\bold{A}(\theta)$, one can get
\begin{equation} \label{Eq:PDM}
a_{pq}(\theta_m) = \exp\{-j\omega[\tau_p(\theta_m)+\tau_q(\theta_m)]\},
\end{equation}
where $\omega$ is the angular velocity. $\tau_p(\theta)$ is the time delay for transmitting between the first antenna and the $p$-th antenna, $\tau_q(\theta)$ is the time delay for receiving between the first antenna and the $q$-th antenna. According to \cite{bekkerman2006target}, the phase delay between the $p$-th and $q$-th antennas can be mathematically expressed as follows:
\begin{equation} \label{Eq:PD}
a_{pq}(\theta) = \exp\{-\frac{2\pi j}{\lambda}(y_p+y_q)\cos{\theta}\},
\end{equation}
where $y_p$ and $y_q$ represent the location of the $p$-th transmitting antenna and $q$-th receiving antenna at the BS, respectively. Following the derivation of $\bold{G}$, the MLE of $\alpha$ and $\theta$ can be obtained by minimizing the negative log-likelihood function:

\begin{figure}[t]
	\centering
	\includegraphics[scale=0.45]{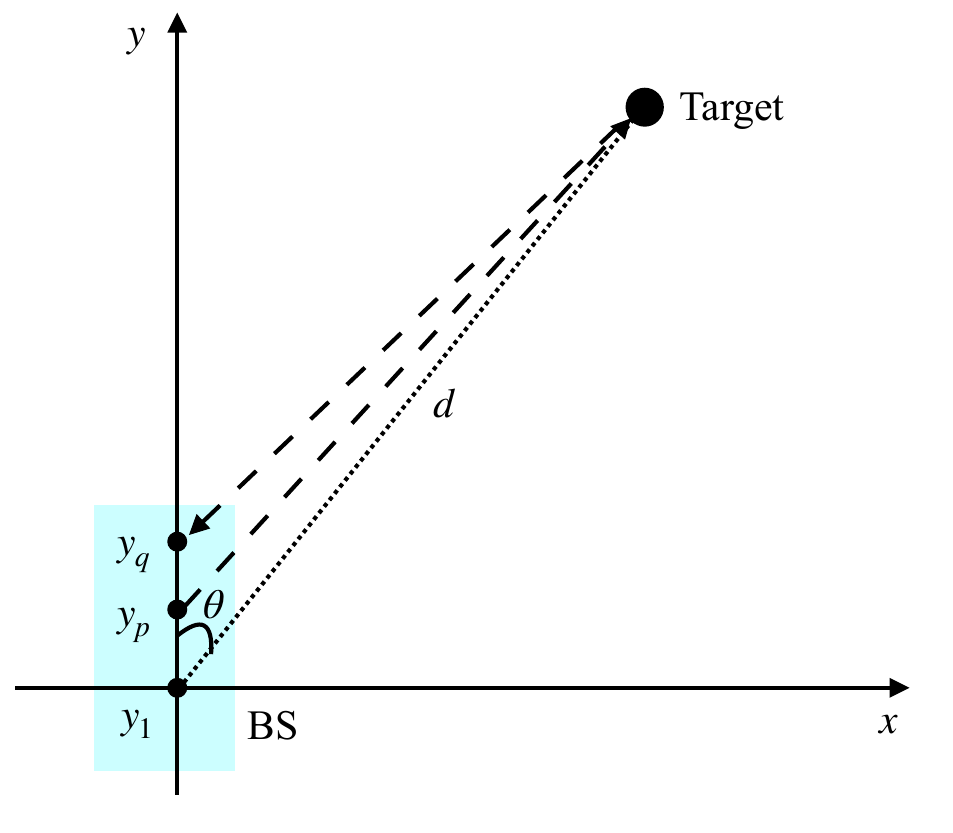}
	\caption{Target Positioning based on ISCPT}
	\label{FigPosition}
\end{figure}

\begin{equation} \label{Eq:MLELoc}
L(\alpha,\theta) \!=\! \text{tr}[(\bold{Y}_R \!-\! \alpha \bold{A}(\theta) \bold{X})^H(\bold{Y}_R \!-\! \alpha \bold{A}(\theta) \bold{X})].
\end{equation}
The derivatives of $L(\alpha,\theta)$ w.r.t. $\alpha$ is
\begin{equation} \label{Eq:Deralpha}
2 \alpha \text{tr}(\bold{X}^H \bold{A}^H(\theta) \bold{A}(\theta) \bold{X}) - 2\text{tr} (\bold{X}^H \bold{A}^H(\theta) \bold{Y}_R).
\end{equation}
According to \eqref{MLEG}, $\hat{\bold{G}} \bold{X} \bold{X}^H = \bold{Y}_R \bold{X}^H$. Setting the derivatives as zero, one can get
\begin{equation} \label{Eq:Estalpha}
\hat{\alpha} = \frac{\text{tr} (\bold{R}_X \bold{A}^H(\theta) \hat{\bold{G}} )}{\text{tr}(\bold{R}_X \bold{A}^H(\theta) \bold{A}(\theta) )}.
\end{equation}
By replacing $\alpha$ with $\hat{\alpha}$ in $L(\alpha,\theta)$, one can get
\begin{equation} \label{Eq:EstA}
L(\theta)= \text{tr}(\bold{Y}_R^H\bold{Y}_R)-\frac{\text{tr}^2(\bold{R}_X\bold{A}^H(\theta)\hat{\bold{G}} )}{\text{tr}(\bold{R}_X\bold{A}^H(\theta)\bold{A}(\theta))}.
\end{equation}
As the angle $\theta$ to be estimated is only relevant with the second item in \eqref{Eq:EstA}, one can get
\begin{equation} \label{Eq:Esttheta}
\hat{\theta} = \arg\max_{\theta} \frac{\text{tr}^2(\bold{R}_X\bold{A}^H(\theta)\hat{\bold{G}} )}{\text{tr}(\bold{R}_X\bold{A}^H(\theta)\bold{A}(\theta))}.
\end{equation}
Note that $\theta$ cannot be expressed analytically in a closed form. Hence, to obtain numerical results, grid search or golden section search techniques can be employed. On the other hand, the distance $d$ between the target and the BS can be estimated following the the free space propagation law \cite{sturm2011waveform}. Based on the estimated parameters (distance $\hat{d}$ and angle $\hat{\theta}$) and its own location $(x,y)$, the BS can estimate the position of the target denoted by $\bold{\hat{p}} = [\hat{x},\hat{y}]^T$ via
\begin{align} 
\hat{x} &= x + \hat{d} \sin \hat{\theta}, \label{Eq:Estx} \\ 
\hat{y} &= y + \hat{d} \cos \hat{\theta}. \label{Eq:Esty}
\end{align}

 
\section{Numerical Results}
In this section, our proposed schemes for sensing both the point target and extended target with separated or co-located ERs and IRs are evaluated in simulations. The performance is further compared with baseline schemes. Unless otherwise specified, the parameters are set as follows:
\begin{itemize}
\item The system comprises 12 IRs and 12 ERs.
\item At the BS, the number of transmit antennas, $N_t$, is set to 16, and the number of receive antennas, $N_r$, is set to 20.
\item The power budget is set to $P = 30$ dBm.
\item The reflection coefficient is set to $\alpha = 0.01$.
\item Both the communication and radar noise powers are $\sigma_R^2 = \sigma_C^2 = 0$ dBm.
\item Both the communication and radar channels are under i.i.d. Rayleigh fading.
\end{itemize}

\subsection{Performance of Beamforming Design for Point Target Sensing}
For point target sensing, two baseline schemes in \cite{liu2020joint} and \cite{stoica2007probing} are considered. The scheme proposed in \cite{liu2020joint} designs beamformer to minimize a weighted sum of beam-pattern loss and mean-squared cross-correlation pattern, which is known as \emph{weighted sum minimization} (WSM) scheme. The scheme proposed in \cite{stoica2007probing} designs beamformer to achieve a desired 3dB main-beam width in the beam-pattern, which is known as \emph{main-beam width} (MBW) scheme. To ensure the fairness in comparison, the energy transmission constraints are also incorporated into the baseline schemes. 

\begin{figure}[ht]
  \centering
  \subfigure[]{
  \label{CRB-SINR-separated}
  \includegraphics[scale=0.4]{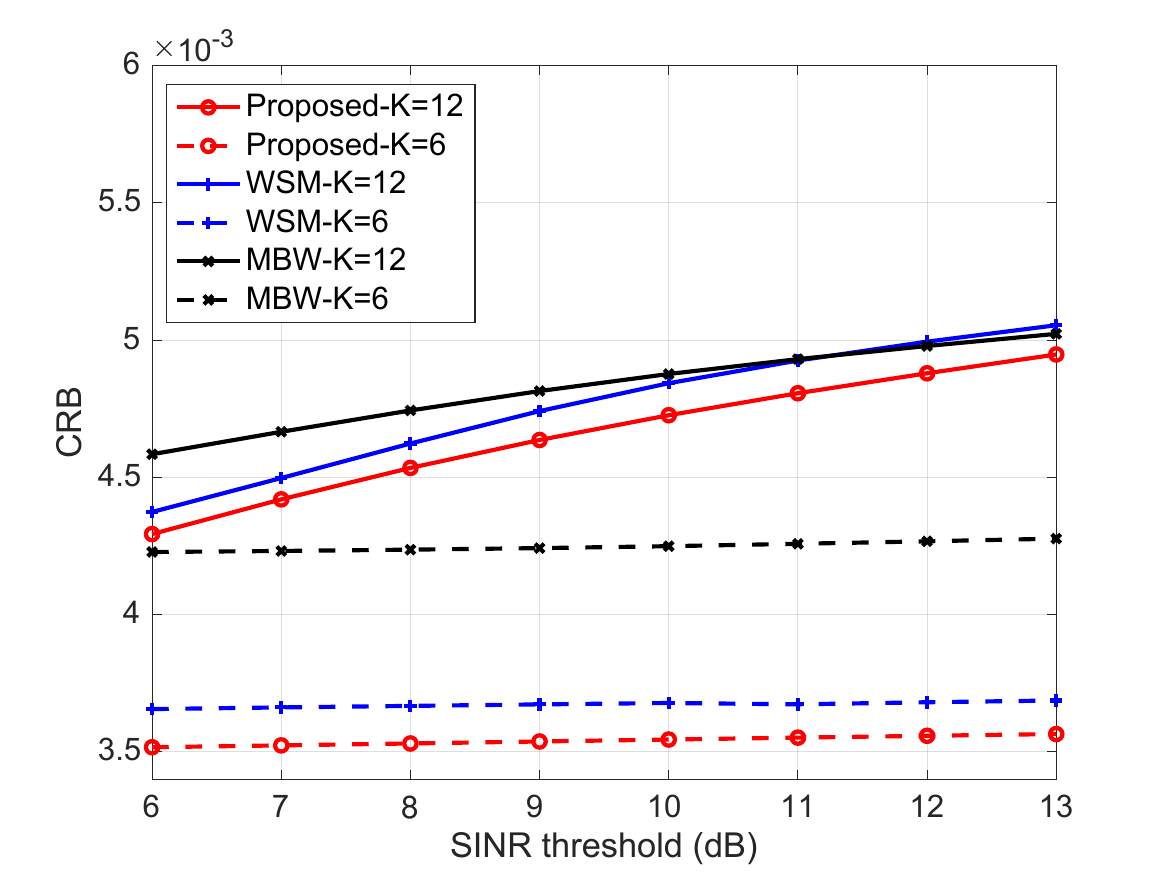}}
  \subfigure[]{
  \label{CRB-SINR-co-located}
  \includegraphics[scale=0.4]{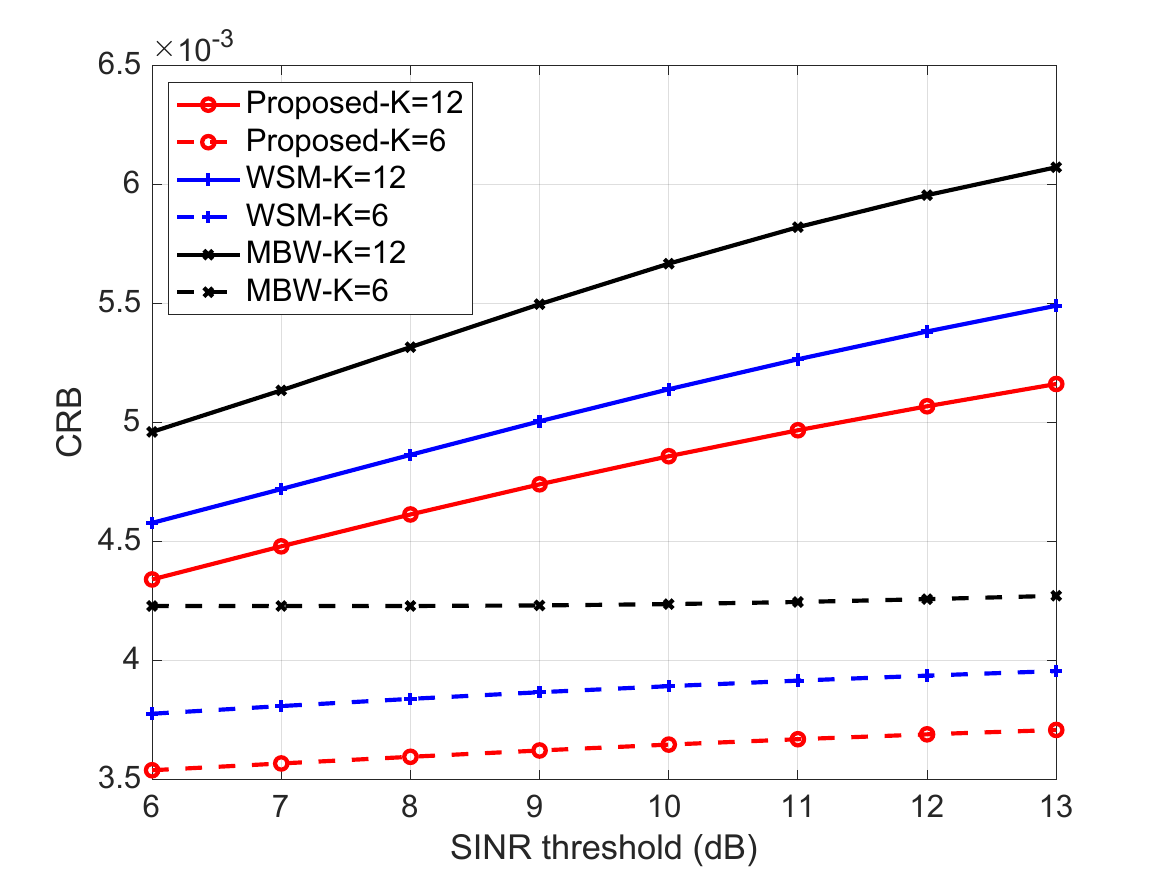}}
  \caption{Sensing CRB versus SINR threshold in a) separated IR/ER case; b) co-located IR/ER case}
  \label{CRB-SINR}
\end{figure}

Fig.~\ref{CRB-SINR} illustrates the sensing CRB versus SINR threshold in both the separated and co-located cases. It can be observed that the sensing CRBs of all schemes increase with growing SINR thresholds, indicating a trade-off between sensing and communication performance. When there are $K = 12$ IRs, the sensing CRB is larger than that of $K = 6$ and increases more rapidly with the increasing SINR threshold. This shows that guaranteeing the communication requirements of more IRs leads to the deteriorated sensing performance. Moreover, our proposed scheme outperforms the two baseline schemes in both the separated and co-located cases.

\begin{figure}[ht]
  \centering
  \subfigure[]{
  \label{CRB-energy-separated}
  \includegraphics[scale=0.4]{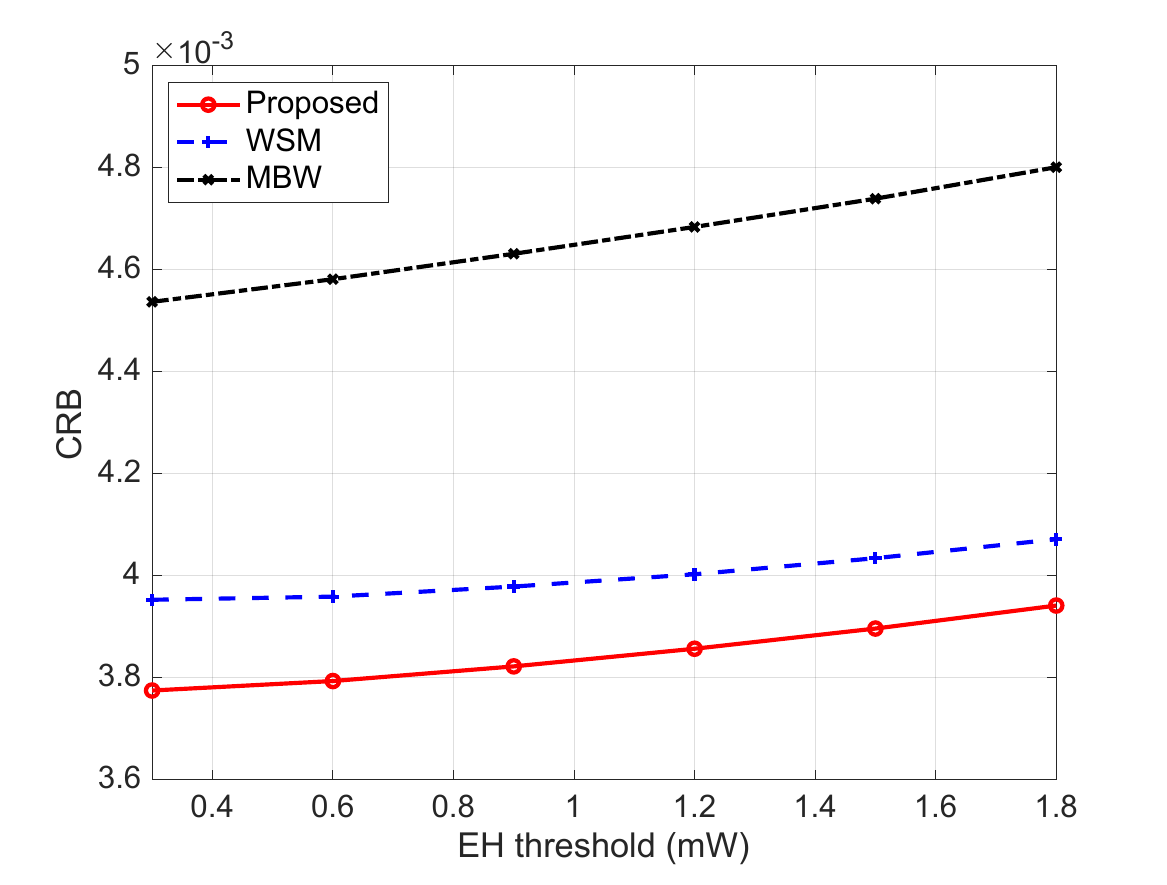}}
  \subfigure[]{
  \label{CRB-energy-co-located}
  \includegraphics[scale=0.4]{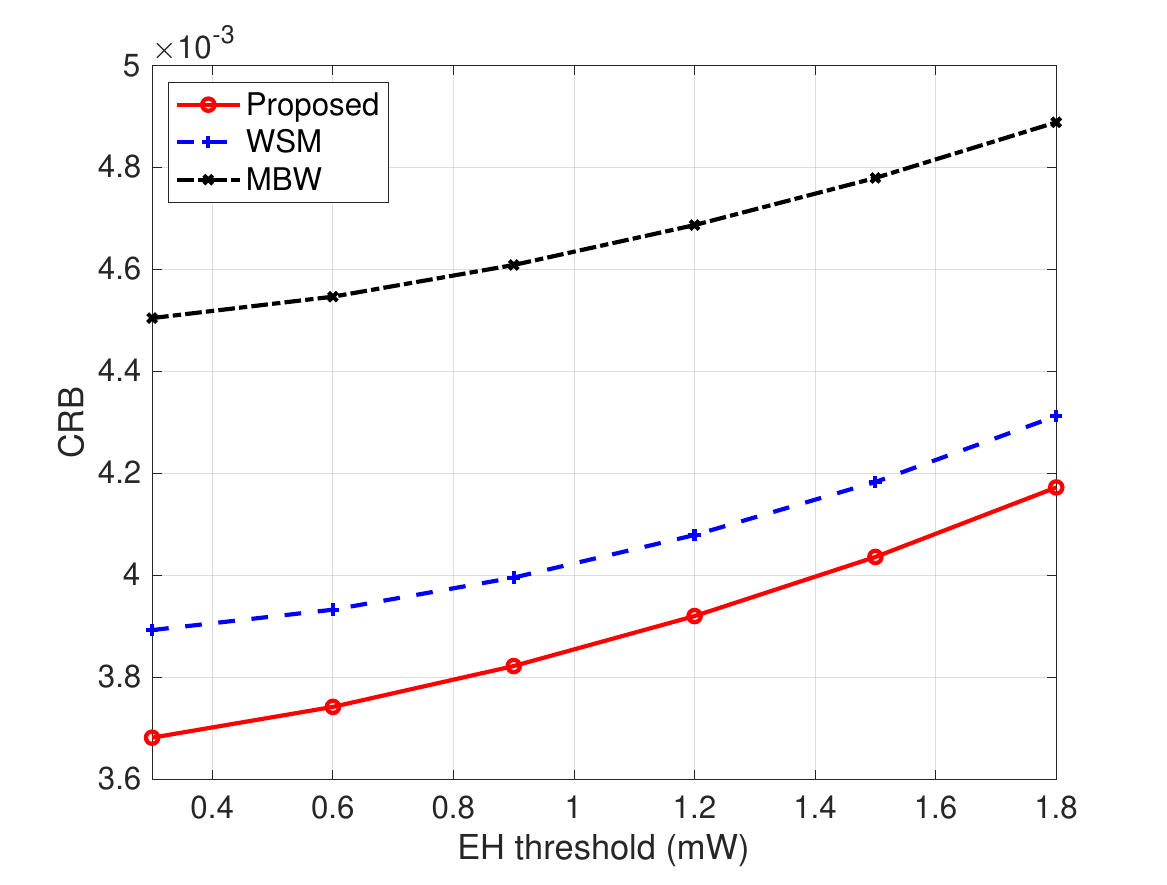}}
  \caption{Sensing CRB versus EH threshold in a) separated IR/ER case; b) co-located IR/ER case}
  \label{CRB-energy}
\end{figure}

Fig.~\ref{CRB-energy} illustrates the sensing CRB versus EH threshold in both the separated and co-located cases. It can be observed that the sensing CRB increases with larger EH threshold, which demonstrates the tradeoff between the sensing and EH performance. Moreover, the sensing CRB increases more rapidly with the increasing EH requirement in the co-located case than in the separated case. This is due to the presence of the PS factor, which amplifies the influence of the EH requirement. 

\begin{figure}[t]
	\centering
	\includegraphics[scale=0.4]{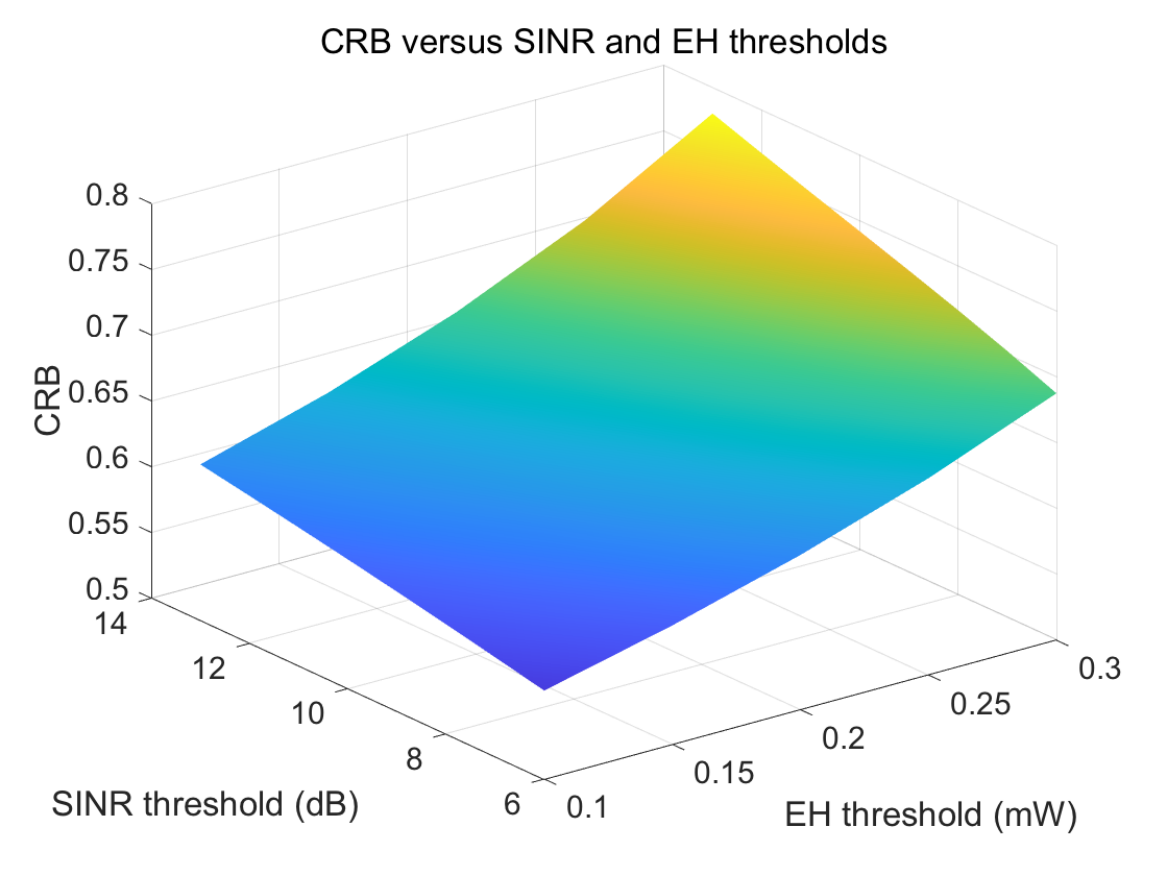}
	\caption{Sensing CRB versus SINR and EH thresholds}
	\label{CRB-3D}
\end{figure}

Fig.~\ref{CRB-3D} shows the sensing CRB versus SINR and EH thresholds in the separated case. It can be observed that sensing CRB increases with the larger SINR and EH thresholds. This illustrates the trade-off among point target sensing, communication, and power transfer performance.

\begin{figure}[t]
	\centering
    \includegraphics[scale=0.4]{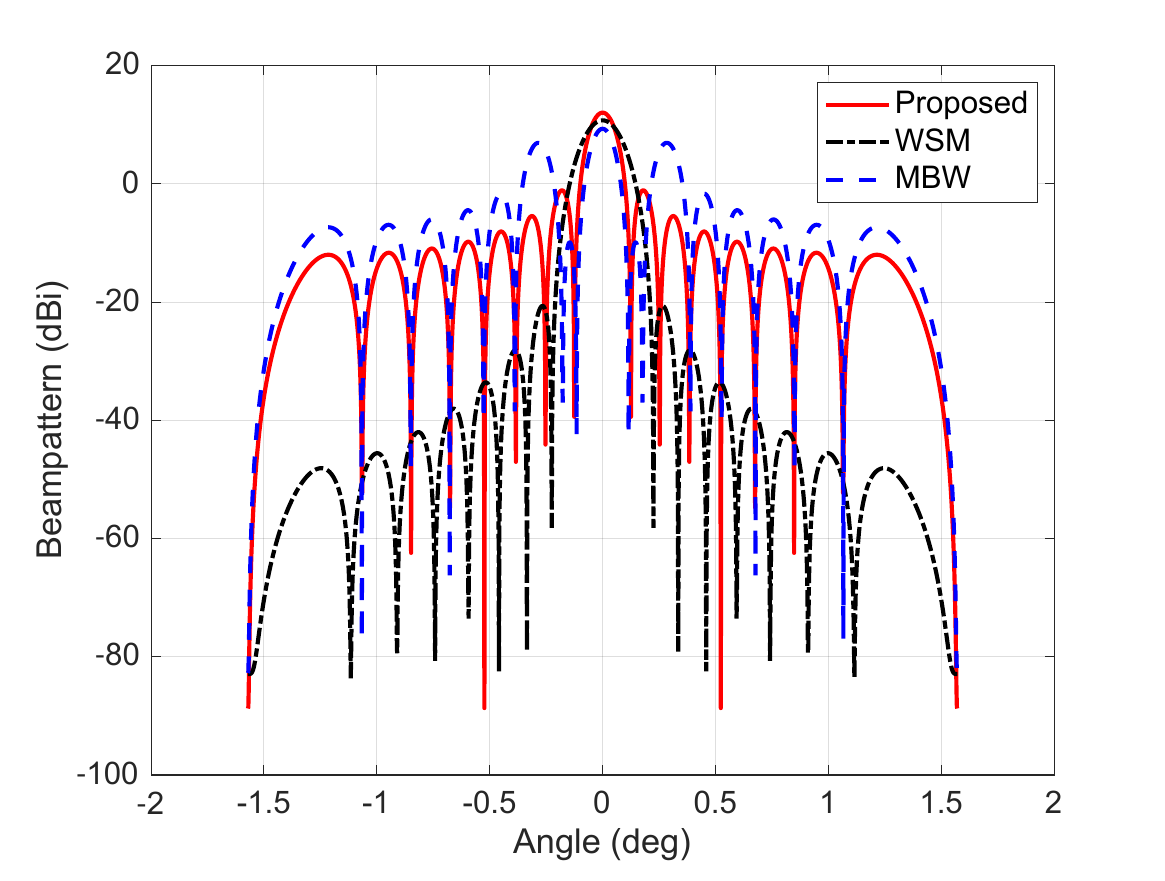}
	\caption{Beampattern for CRB minimization in the separated case}
	\label{Beampattern}
\end{figure}

The beam pattern of our proposed scheme compared with two baselines in the separated case is shown in Fig.~\ref{Beampattern}. It can be observed that all beamformers obtained by these three schemes accurately steer the main lobe towards the degree of 0. Moreover, there exists random fluctuations in the side-lobe regions of all the obtained beampatterns, which depends on the EH and SINR constraints. Among all three schemes, our proposed one exhibits the highest power radiation towards the target angle. 

\subsection{Performance of Beamforming Design for Extended Target Sensing}
The performance of extended target sensing is illustrated in this sub-section. Two beamforming designs based on the convex relaxation bound and eigenvalue decomposition are adopted for performance comparison. The convex relaxation bound refers to the solution of the optimization problem after SDR, whose rank might be more than 1. Eigenvalue decomposition recovers the rank-1 solution by extracting the largest eigenvalue and its corresponding eigenvector of the covariance matrix. Each scheme solves the SDR problem to obtain the covariance matrix and the only difference lies in the process of extracting the rank-1 solution. 

\begin{figure}[ht]
  \centering
  \subfigure[]{
  \label{MSE-SINR-separated}
  \includegraphics[scale=0.4]{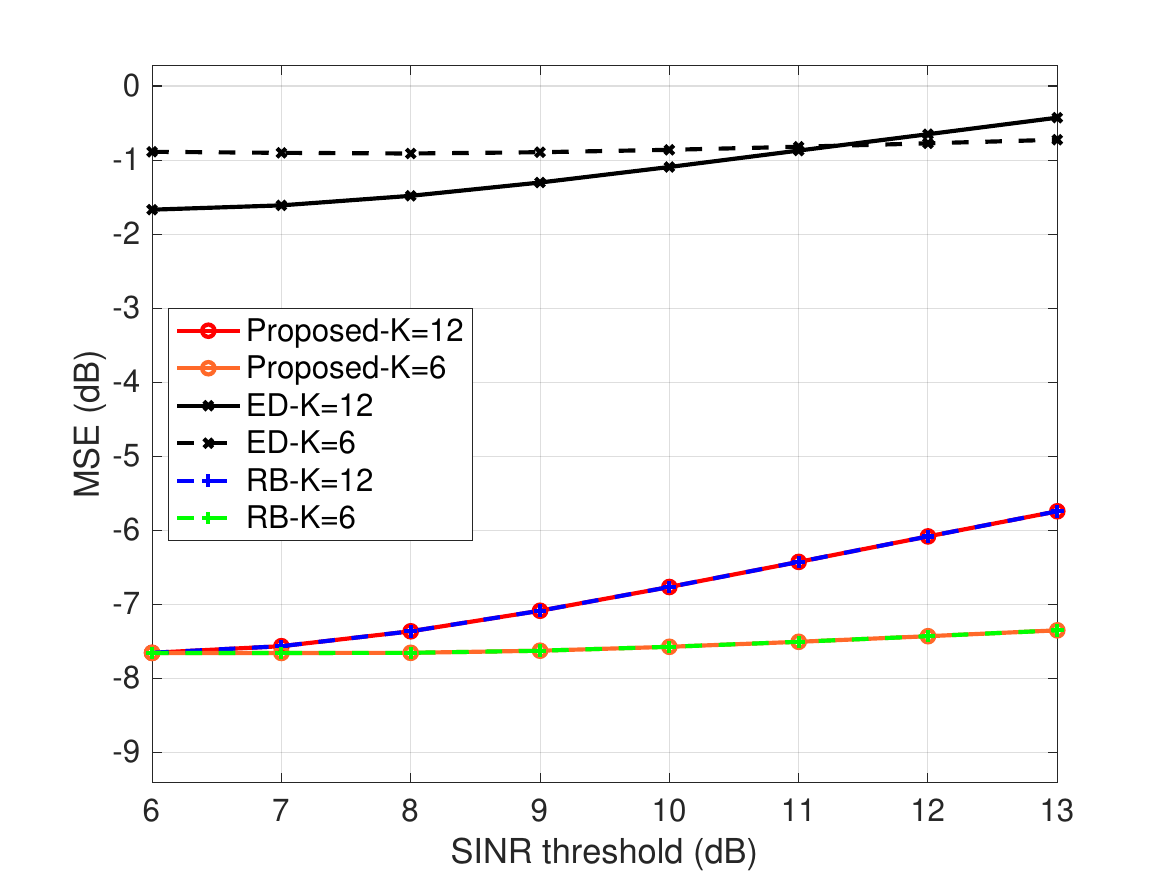}}
  \subfigure[]{
  \label{MSE-SINR-colocated}
  \includegraphics[scale=0.4]{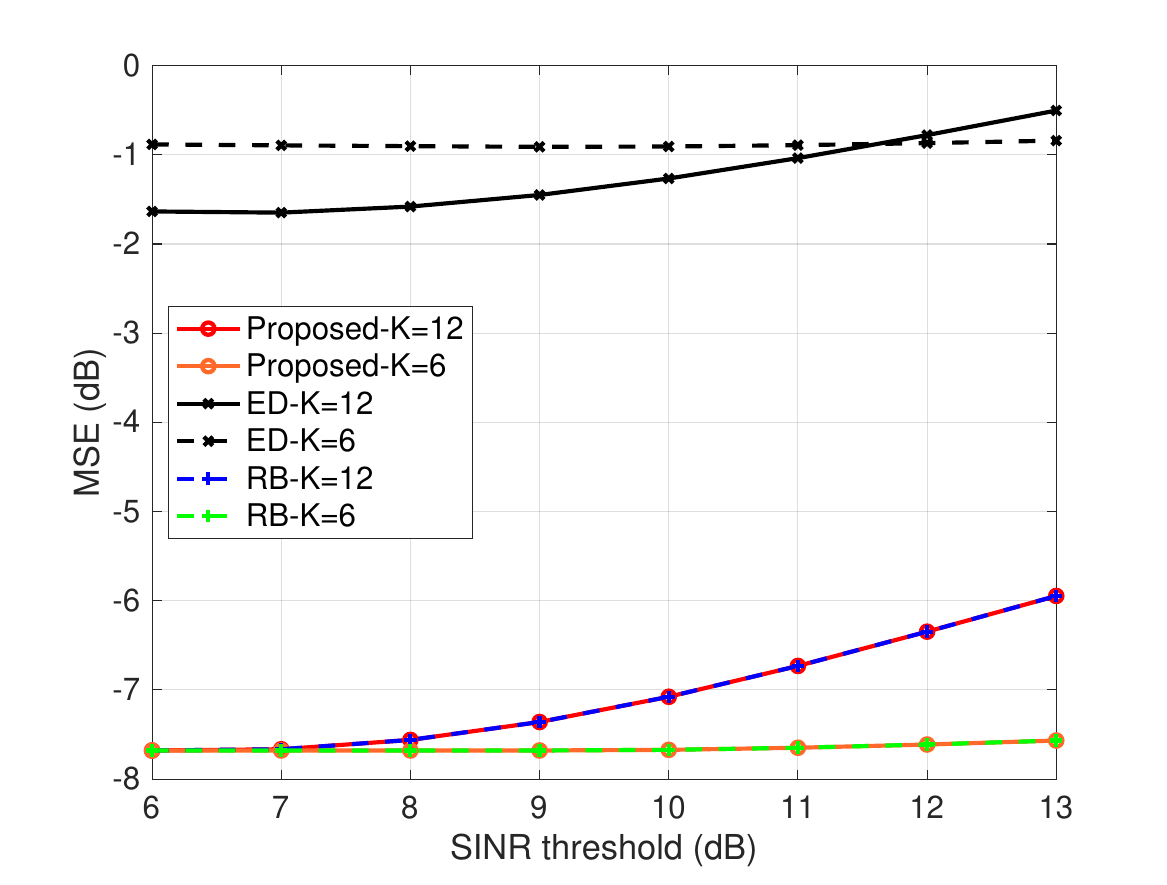}}
  \caption{Sensing MSE versus SINR threshold in a) separated IR/ER case; b) co-located IR/ER case}
  \label{MSE-SINR}
\end{figure}

Fig.~\ref{MSE-SINR} depicts the sensing MSE versus the SINR threshold in both the separated and co-located cases. It can be observed that the sensing MSE increases with larger SINR threshold, which demonstrates the tradeoff between the sensing and communication performance. One can also observe that our proposed scheme significantly outperforms the eigenvalue decomposition method and even achieves the convex relaxation bound, which verifies that the rank-1 solution can be obtained without altering the optimal objective value. Moreover, the sensing MSE at $K = 12$ is larger that that at $K = 6$ in our proposed scheme, as the sensing performance is sacrificed for guaranteeing the communication requirements of more IRs.

\begin{figure}[ht]
  \centering
  \subfigure[]{
  \label{MSE-Energy-colocated}
  \includegraphics[scale=0.4]{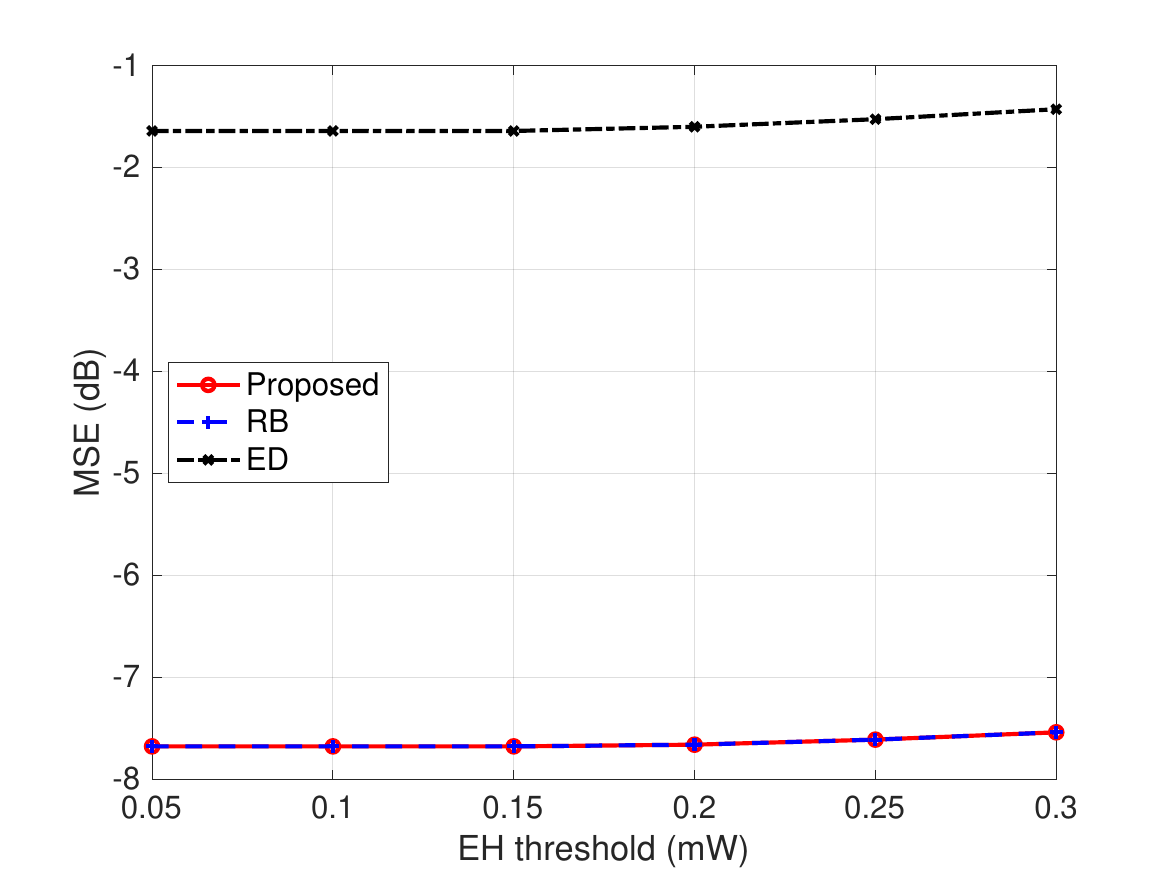}}
  \subfigure[]{
  \label{MSE-Energy-separated}
  \includegraphics[scale=0.4]{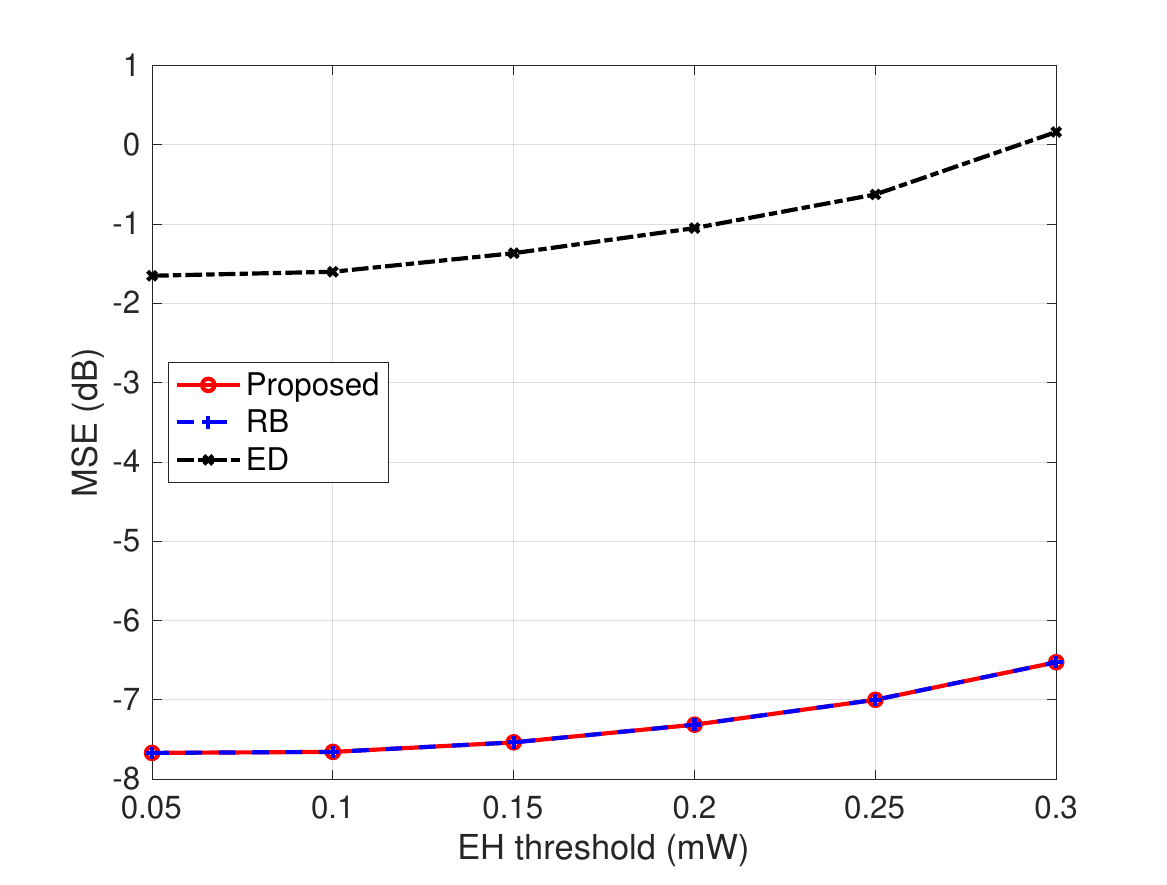}}
  \caption{Sensing MSE versus ER threshold in a) separated IR/ER case; b) co-located IR/ER case}
  \label{MSE-Energy}
\end{figure}

Fig.~\ref{MSE-Energy} shows the sensing MSE versus the EH threshold in both the separated and co-located cases. It can be observed that the sensing MSE increases with larger EH threshold, which demonstrates the tradeoff between the sensing and EH performance. One can also observe that in the co-located case, the sensing MSE increases more rapidly with the larger EH threshold compared to the separated case, which is also due to the existence of the PS factor.

\begin{figure}[t]
	\centering
    \includegraphics[scale=0.4]{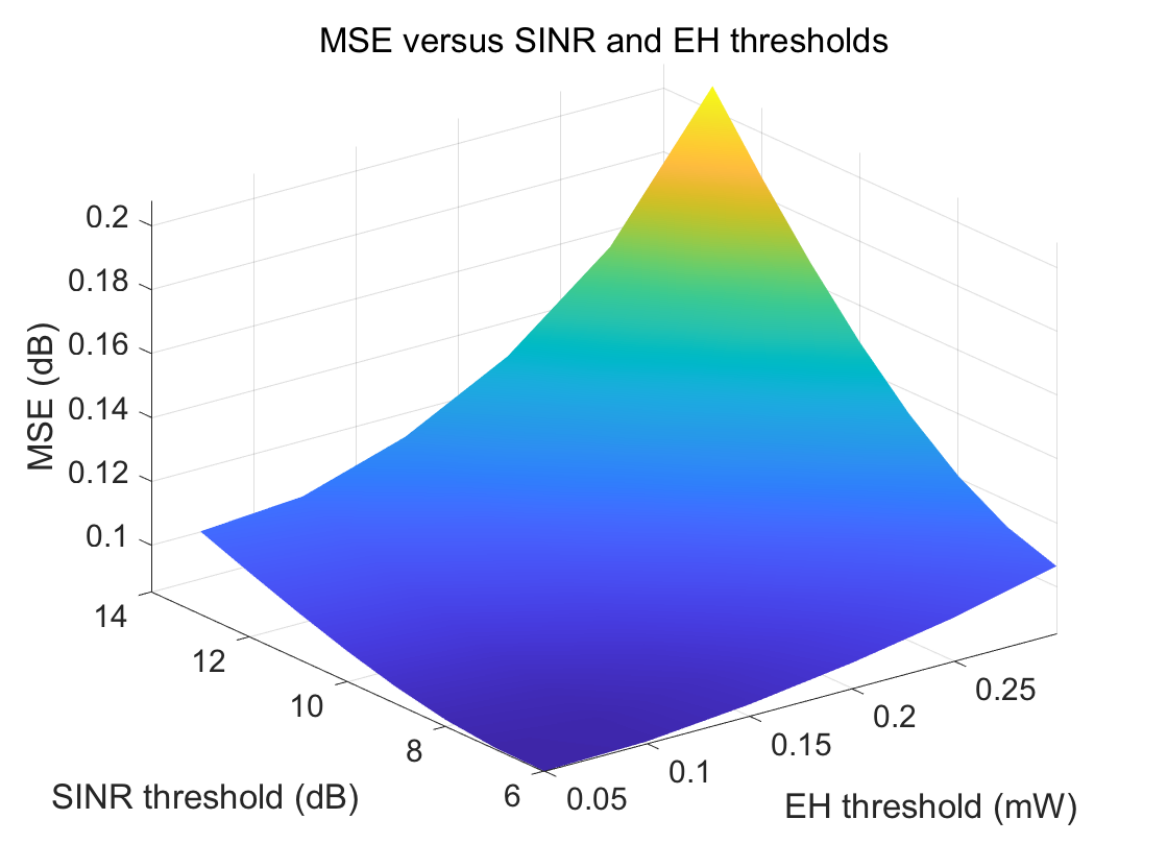}
	\caption{Sensing MSE versus SINR and EH thresholds}
	\label{MSE-3D}
\end{figure}

Fig.~\ref{MSE-3D} shows the sensing MSE versus SINR and EH thresholds  in the separated case. It can be observed that sensing MSE increases with the larger SINR and EH thresholds. This illustrates the trade-off among extended target sensing, communication, and power transfer performance.

\begin{figure}[t]
	\centering
    \includegraphics[scale=0.4]{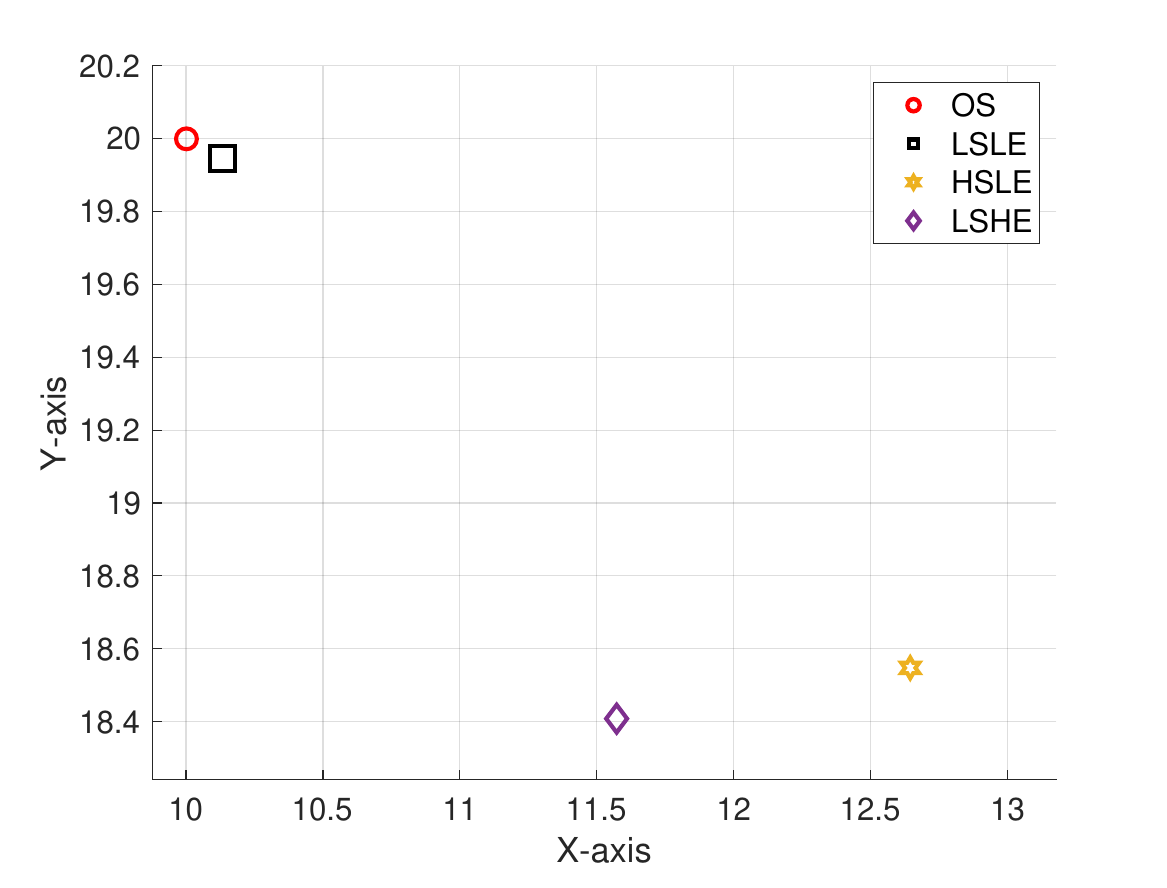}
	\caption{Performance of target positioning based on ISCPT}
	\label{FigTarget}
\end{figure}

\subsection{Target Positioning based on ISCPT}
The performance of target positioning based on ISCPT is illustrated in Fig.~\ref{FigTarget}. The target to be estimated locates at  $(10,20)$, and the BS locates at $(0,0)$. There are $N_{tx} = 16$ transmitting antennas and $N_{rx} = 20$ receiving antennas at the BS, distributed along the y-axis. Four schemes are designed for the performance comparison. The red circle point represents the target location estimated by the radar signals only for sensing, which is known as the \emph{only sensing} (OS) scheme. The black rectangular point represents the target location estimated by the ISCPT signals with low SINR (8 dB) and low EH (0.1 mW) thresholds, which is known as the \emph{low SINR low EH} (LSLE) scheme. The purple diamond point represents the target location estimated by the ISCPT signals with high SINR (12 dB) and low EH (0.1 mW) thresholds, which is known as the \emph{high SINR low EH} (HSLE) scheme. The yellow hexagram point represents the target location estimated by the ISCPT signals with low SINR (8 dB) and high EH (0.2 mW) thresholds, which is known as the \emph{low SINR high EH} (LSHE) scheme.

It can be observed that the OS scheme can perfectly estimate the position of the target, as the signals are purely used for sensing. The LSLE scheme can also realize accurate target positioning, which demonstrates the possibility of achieving high sensing performance while guaranteeing the basic requirements of communication and power transfer. However, when the SINR threshold or the EH threshold is relatively high, the target positioning  in HSLE and LSHE schemes will be less accurate, indicating the tradeoff among the sensing, communication, and power transfer performance.

\section{Conclusion}
In this paper, a multi-user MIMO ISCPT system is designed, where a BS equipped with multiple antennas delivers message to multiple IRs, transfers power to multiple ERs, and senses a target simultaneously. Beamforming designs are optimized to minimize the sensing CRB while satisfying SINR and EH constraints for point target sensing. The non-convex optimization problem is solved using techniques such as Schur complement transformation and rank reduction. When the IRs and ERs are co-located, joint optimization of PS factors and beamformers balances communication and power transfer performance. For extended target sensing, the objective is to minimize sensing MSE, and the problems are solved by leveraging the rank property of the covariance matrix. The performance of ISCPT is illustrated through the investigation of the target positioning problem.
This work contributes to the promising new research area of ISCPT and many interesting follow-up research issues warrant further investigation, such as device scheduling and multiple BSs cooperation.



\bibliographystyle{IEEEtran}

\end{document}